\documentclass[11pt]{article}

\usepackage[margin=1in]{geometry}
\usepackage{amsmath,amssymb,amsthm}
\usepackage{bm}
\usepackage{booktabs}
\usepackage{enumitem}
\usepackage{hyperref}
\usepackage{natbib}
\usepackage{graphicx}
\usepackage[section]{placeins}
\usepackage{multirow}
\usepackage{array}
\usepackage{longtable}
\usepackage{mathtools}
\usepackage{url}

\graphicspath{{psnb_sim/output/figures/}}

\hypersetup{
  colorlinks=true,
  linkcolor=blue,
  citecolor=blue,
  urlcolor=blue
}

\newcommand{\E}{\mathbb{E}}
\newcommand{\Prob}{\mathbb{P}}
\newcommand{\Ind}{\mathbb{I}}
\newcommand{\R}{\mathbb{R}}
\DeclareMathOperator{\Var}{Var}

\newtheorem{definition}{Definition}
\newtheorem{assumption}{Assumption}
\newtheorem{theorem}{Theorem}
\newtheorem{proposition}{Proposition}
\newtheorem{corollary}{Corollary}

\title{
Priority-Standardized Net Benefit:\
A Stage-Normalized Estimand for Hierarchical Composite Endpoints
}

\author{
David McCoy, John Leopold, Shirley Galbiati, Minhthien Vu, and Bonnie Zhang \\
\vspace{2mm}
\textit{Edwards Lifesciences}
}

\date{\today}

\begin{document}
\maketitle

\begin{abstract}
Hierarchical composite endpoints analyzed with win statistics are increasingly used
when outcomes differ in clinical importance and hard events are too rare to support
a single-component primary endpoint. Their main practical appeal is that the analysis
respects a prespecified priority order. Their main practical vulnerability is less often
stated: standard win summaries aggregate layer-specific information using \emph{reach
probabilities}---the fraction of treated--control pairs that remain tied long enough
to be compared at each layer. Consequently, when upper layers are rare or highly tied,
a frequently reached last layer can numerically dominate the composite even if it is
lowest priority and potentially more vulnerable to bias or missingness, such as an
open-label patient-reported outcome.

We propose the \textbf{Priority-Standardized Net Benefit} (PSNB), an estimand that
decomposes a hierarchical comparison into stage-conditional net benefits and recombines
them using a prespecified \emph{priority/credibility charter} rather than data-determined
reach weights. The paper makes three contributions. First, it identifies the main
methodological gap as one of \emph{across-layer aggregation}, not within-layer comparison,
and shows that fixed weighted win--loss statistics remain mechanically reach-weighted.
Second, it develops an influence-function-based estimator and large-sample inference for
PSNB, together with a ratio-scale companion, the \emph{Priority-Standardized Win Ratio}
(PSWR). Third, it introduces practical design tools---a layer influence cap, tipping-point
analysis for the last-layer weight, and charter-envelope sensitivity analysis---that help
trial teams document the evidentiary role of each hierarchy layer in the protocol or
statistical analysis plan before unblinding.

Simulations are organized around the properties that distinguish PSNB from standard win
summaries. They confirm nominal type I error at realistic trial sample sizes, show that
PSNB is approximately invariant to large changes in reach when stage-conditional effects
are held fixed, and demonstrate that upstream tie rules can make standard hierarchical
summaries almost entirely later-layer-driven. They further show that sensitivity to
late-layer bias and missingness is not automatically removed by PSNB, but can be traced
to the chosen charter and is largely invariant to reach: smaller later-layer weights
attenuate that sensitivity, whereas larger later-layer weights can
increase it relative to standard reach-weighted summaries. Finally, power varies
predictably along a design-stage efficiency--robustness frontier and along a
simulation-based sample-size curve. In the design scenarios studied here, at a sample
size where the standard Win Ratio has approximately $90\%$ power under broad benefit,
PSNB with the baseline charters keeps pace; under final-layer-dominated benefit, power
depends strongly on the permitted Layer-3 charter weight. Taken together, these results
support PSNB as an
estimand-level modification of
hierarchical composites that preserves prioritized pairwise comparisons while replacing
implicit reach-weighting with a design choice fixed before unblinding.
\end{abstract}


\section{Introduction}

Hierarchical composite endpoints are increasingly used in randomized trials when
outcomes differ in clinical importance and a single time-to-first-event endpoint
would either discard clinically relevant information or treat unlike outcomes as
exchangeable. A common setting is a hierarchy in which death is prioritized above
a non-fatal clinical event and a later patient-reported or health-status outcome.
Investigators often choose such hierarchies because they want the primary analysis
to respect that clinical ordering.

The standard win-ratio framework of \citet{Pocock2012}, building on the
Finkelstein--Schoenfeld construction \citep{FinkelsteinSchoenfeld1999} and the
more general generalized pairwise comparison framework of \citet{Buyse2010},
addresses this by comparing treated and control participants according to a
prespecified priority order. The first layer at which a treated--control pair
differs determines a win, loss, or tie. This is clinically appealing because it
respects the hierarchy at the level of pairwise comparison.

Yet a hierarchical comparison rule is not the same as a hierarchy-respecting
summary measure. Standard win summaries aggregate stage-specific information using
the empirical \emph{reach} distribution: a layer contributes to the overall result
only insofar as treated--control pairs remain tied on higher-priority layers and
therefore reach it. Consequently, when upper layers are rare or highly tied,
later layers are reached often and can contribute disproportionately to the
headline composite summary. The hierarchy stated in the trial objective may
therefore prioritize hard outcomes, while the numerical value of the primary
summary is driven by whichever layer pairs most often reach.

This issue is especially important when a frequently reached later layer is more
vulnerable than earlier layers to contextual effects, bias, or missingness.
Examples include symptom response, health-status change, or other patient-reported
outcomes collected without blinding. The concern is not that such outcomes are
unimportant. On the contrary, they are often central to the clinical question.
The concern is that a \emph{hierarchical} composite can become, in practice, a
\emph{de facto later-layer endpoint} without showing that dependence in the trial
report.
A single published win ratio may therefore be interpreted as if benefit were
broadly distributed across the hierarchy when, numerically, the result is
sustained largely by one frequently reached layer.

Existing methodological work has improved several important aspects of
hierarchical composites. Some methods modify the pairwise comparison rule or
score, for example through fixed weighting, thresholds, or clinically meaningful
margins \citep{Luo2017WeightedWinLoss,Dong2020,Mou2024thresholds}. Other work
clarifies the pairwise estimand targeted by win statistics under censoring and
pairing schemes \citep{Mao2024EstimandWR,EvenJosse2025CausalWR}. Less attention
has been given to a separate design question: once those pairwise comparisons
have been defined, how should stage-specific information be recombined
\emph{across} the hierarchy into the primary summary measure? We view that
recombination as part of the summary-measure attribute of the estimand under
ICH E9(R1) \citep{ICHE9R1}.

This paper addresses that question. We propose the
\emph{Priority-Standardized Net Benefit} (PSNB), an estimand that leaves the
within-layer pairwise comparison rule unchanged but replaces data-determined reach
weights with a prespecified \emph{priority/credibility charter}. Under PSNB, the
influence of each layer on the primary summary is fixed before unblinding rather
than left to how often pairs happen to reach that layer. For readers who prefer a ratio-scale summary,
we also define a companion \emph{Priority-Standardized Win Ratio} (PSWR).

\paragraph{Contributions.}
The paper makes three contributions.
\begin{enumerate}[leftmargin=*,label=(\roman*)]
\item It identifies across-layer aggregation as an under-specified estimand
choice in hierarchical composite analysis and shows that fixed weighted
win--loss summaries do not remove the problem because their layer contributions
remain mechanically reach-weighted.
\item It defines PSNB as a stage-standardized estimand, develops a two-sample
$U$-statistic estimator with influence-function-based large-sample inference, and
derives a simple layer influence cap showing that the maximum contribution of any
layer is bounded by its charter weight.
\item It introduces prespecifiable design and reporting tools---a tipping-point
analysis for the weight on a later layer and a charter-envelope sensitivity
analysis over an admissible weight set---that can be incorporated directly into
a statistical analysis plan before unblinding.
\end{enumerate}

The paper is written as an estimand and trial-design contribution rather than as
a commentary on any specific published trial. All illustrations use simulation
scenarios calibrated to features common in contemporary hierarchical composites:
rare upper-layer hard events, an intermediate objective layer, and a frequently
reached later layer that may be more vulnerable to bias or missingness than the
layers above it. The intended use of PSNB is prospective: to help define,
justify, and transparently report the role of each hierarchy layer before
unblinding.

\section{Related work and methodological positioning}

\subsection{Standard win statistics and why reach matters}

Under generalized pairwise comparison (GPC), a treated outcome and a control outcome
are compared according to a prioritized rule, and the first non-tied layer determines
the pairwise result. The best-known summaries are the \emph{win ratio},
$\mathrm{WR}=\Prob(\text{win})/\Prob(\text{loss})$, the \emph{net benefit} or
\emph{win difference}, $\Delta=\Prob(\text{win})-\Prob(\text{loss})$, and the
\emph{win odds}, which incorporates ties symmetrically
\citep{Buyse2010,Pocock2012,Dong2020,Brunner2021}. A substantial literature now
supports the use of these summaries through large-sample inference, sample-size
formulae, and practitioner-oriented guidance on interpretation and reporting
\citep{BebuLachin2016,Mao2022SampleSize,Pocock2024Lessons}.

These summaries are useful and often clinically intuitive, but they aggregate
stage-specific information through the empirical flow of pairs across the hierarchy.
A layer contributes to the overall summary only to the extent that treated--control
pairs remain tied on higher-priority layers and therefore \emph{reach} it.
Consequently, the final numerical value depends not only on the stated priority
order, but also on how often pairs are tied at higher layers, which in turn depends
on event rarity, censoring, threshold choices, and the noisiness of earlier layers.
These dependencies are not flaws of the estimators; they are consequences of how the
standard summary functionals themselves are defined.

\subsection{Three adjacent strands of methodological development}

Existing work has improved hierarchical composites along three broad lines.

First, some methods modify the \emph{pairwise score or stage-specific comparison
rule}. Examples include fixed weighted win--loss summaries, which attach
pre-specified weights to wins and losses resolved at different stages
\citep{Luo2017WeightedWinLoss}, and threshold- or margin-based rules for continuous
components, which redefine what constitutes a clinically meaningful win within a
stage \citep{Mou2024thresholds}. These approaches can improve stage-level
interpretability, but they do not remove the dependence of the overall composite
on how frequently a stage is reached.

Second, some methods refine the \emph{summary functional itself}, especially when
ties are common. \citet{Dong2020} discuss interpretation of the win ratio and the
role of ties, while \citet{Brunner2021} propose the win odds as a tie-aware
alternative. These developments improve communication and inference when ties are
substantial, but they still aggregate information through unconditional pairwise
win, loss, and tie frequencies induced by the hierarchy.

Third, other work has focused on \emph{inference, design, and estimands} for
existing win statistics. \citet{BebuLachin2016} derive large-sample inference using
two-sample $U$-statistic arguments, \citet{Mao2022SampleSize} develop sample-size
formulae for general win-ratio analyses, and recent work clarifies the estimand and
causal target of win statistics under censoring and pairing schemes
\citep{Mao2024EstimandWR,EvenJosse2025CausalWR}. Recent practitioner-facing reviews
have also emphasized that interpretation can become difficult when lower-priority
tiers account for much of the numerical signal \citep{Pocock2024Lessons}.

PSNB is complementary to all three strands. It does not replace the within-stage
comparison rule, does not change how ties are defined or summarized, and does not
alter the basic two-sample pairwise estimand perspective. Instead, it changes the
\emph{aggregation map}: it replaces reach-weighted aggregation of stage information
with aggregation based on a prespecified charter.

\subsection{Positioning PSNB among adjacent summaries}

Table~\ref{tab:positioning} summarizes the specific methodological niche of PSNB.
The key distinction is not the pairwise notion of a win, but the \emph{across-layer
aggregation rule}.

\begin{table}[ht]
\centering
\caption{Conceptual positioning of PSNB relative to adjacent hierarchical summaries.
The distinguishing feature is the across-layer aggregation rule.}
\label{tab:positioning}
\begin{tabular}{@{}p{3.8cm}p{4.4cm}p{3.6cm}p{3.1cm}@{}}
\toprule
Method family & What is aggregated across the hierarchy? & Does layer influence depend on reach? & Is maximum layer influence fixed \emph{a priori}? \\
\midrule
Standard win summaries (WR, net benefit, win odds) &
Unconditional pairwise win/loss/tie probabilities resolved across layers &
Yes &
No \\
Fixed weighted win--loss \citep{Luo2017WeightedWinLoss} &
Weighted unconditional pairwise stage scores &
Yes &
No \\
PSNB (this paper) &
Weighted \emph{stage-conditional} net benefits &
No &
Yes \\
\bottomrule
\end{tabular}

\medskip
\begin{minipage}{0.97\linewidth}
\footnotesize
\emph{Note:} PSWR is a ratio-scale companion to PSNB that uses the same
prespecified charter and the same reach-independent aggregation principle.
It is not listed as a separate methodological family because it is intended as
a secondary companion summary rather than a distinct primary contribution.
\end{minipage}
\end{table}

\section{Setup and notation}

\subsection{Potential outcomes, pairwise comparisons, and reach}

Consider a two-arm randomized trial with treatment indicator $Z\in\{0,1\}$.
Each participant has a $K$-component outcome vector
\[
Y=\bigl(Y^{(1)},\ldots,Y^{(K)}\bigr),
\]
ordered from highest to lowest clinical priority.
Let $Y(1)$ and $Y(0)$ denote the potential outcome vectors under treatment and control.

Following the generalized pairwise comparison (GPC) perspective \citep{Buyse2010},
define two independent draws
\[
Y_1 \sim \mathcal{L}\{Y(1)\},
\qquad
Y_0 \sim \mathcal{L}\{Y(0)\},
\]
interpretable as a random treated participant and a random control participant from
the target population. Under consistency, no interference, and randomization
(formalized in Assumption~\ref{ass:rand} below), these arm-specific distributions are
identified by the observed outcome distributions in the two randomized arms.

For each layer $k\in\{1,\ldots,K\}$, define a prespecified pairwise comparison function
\[
c_k(y_1,y_0)\in\{-1,0,+1\},
\]
where $+1$ means treatment is favored at layer $k$, $-1$ means control is favored,
and $0$ denotes a tie, including clinically irrelevant differences under the chosen
comparison rule.

Typical examples include:
\begin{itemize}[leftmargin=*]
\item \textbf{Time-to-event layers:} later event time is better through a fixed horizon
$\tau$; pairs that cannot be informatively ordered through $\tau$ under the chosen
censoring strategy are treated according to that strategy, for example as ties in a
simple restricted-horizon comparison.
\item \textbf{Counts or recurrent events:} fewer events through $\tau$ is better,
optionally up to a prespecified margin; pairs within the margin are tied.
\item \textbf{Continuous or PRO layers:} greater improvement is better when the
treated--control difference exceeds a clinically meaningful threshold; otherwise the
pair is tied.
\end{itemize}

For time-to-event layers, the choice of censoring rule is important but conceptually
separate from the contribution of the present paper. PSNB does not prescribe how
censoring should be handled; it can be paired with any prespecified comparison rule for
censored outcomes, including simple fixed-horizon rules or censoring-adjusted approaches
such as those discussed by \citet{Mao2024EstimandWR}. Our focus is what happens
\emph{after} the stage-specific pairwise comparisons have been defined.

Define the stage-reach indicator
\[
R_k(y_1,y_0)=\prod_{m=1}^{k-1}\Ind\{c_m(y_1,y_0)=0\},
\]
with the convention $R_1\equiv 1$. Thus $R_k=1$ if and only if the pair is tied on all
higher-priority layers and therefore reaches layer $k$.

The standard hierarchical comparison score is
\[
\phi(y_1,y_0)=\sum_{k=1}^K R_k(y_1,y_0)c_k(y_1,y_0).
\]
Because $R_k=1$ only when all higher-priority layers are tied, at most one term in this
sum can be nonzero; equivalently, only the first non-tied layer contributes. Therefore
\[
\phi(y_1,y_0)\in\{-1,0,+1\}.
\]

\subsection{Standard win summaries as reach-weighted functionals}

Let
\[
W=\Prob\{\phi(Y_1,Y_0)=+1\},
\qquad
L=\Prob\{\phi(Y_1,Y_0)=-1\},
\qquad
T=\Prob\{\phi(Y_1,Y_0)=0\}.
\]
The standard net benefit is $\Delta=W-L$, and the standard win ratio is
$\mathrm{WR}=W/L$ when $L>0$.

For each stage $k$, define the reach probability
\[
r_k=\Prob\{R_k(Y_1,Y_0)=1\},
\]
and the stage-conditional win, loss, and tie probabilities
\[
w_k=\Prob\{c_k(Y_1,Y_0)=+1\mid R_k=1\},
\qquad
\ell_k=\Prob\{c_k(Y_1,Y_0)=-1\mid R_k=1\},
\qquad
t_k=1-w_k-\ell_k.
\]
Define the stage-conditional net benefit
\[
\Delta_k
=
w_k-\ell_k
=
\E\!\left[c_k(Y_1,Y_0)\mid R_k(Y_1,Y_0)=1\right].
\]

Then the unconditional win and loss probabilities decompose as
\[
W=\sum_{k=1}^K r_k w_k,
\qquad
L=\sum_{k=1}^K r_k \ell_k,
\]
and the standard net benefit decomposes as
\begin{equation}\label{eq:standard_reach_weight}
\Delta
=
W-L
=
\sum_{k=1}^K r_k\,(w_k-\ell_k)
=
\sum_{k=1}^K r_k \Delta_k.
\end{equation}

Equation~\eqref{eq:standard_reach_weight} is the central structural fact of the paper.
Standard hierarchical net benefit is a weighted average of stage-conditional effects,
but the weights are reach probabilities determined by the joint outcome distribution and
the chosen tie/censoring conventions, not by the stated clinical priorities alone.

\paragraph{Interpretation.}
The event $R_k=1$ is a property of a \emph{pairwise comparison}, not a causal
subpopulation of individuals. It does not identify patients who ``survive to stage $k$'';
rather, it identifies treated--control pairs that are tied on all higher-priority layers.
Accordingly, $\Delta_k$ should be interpreted as follows:
\emph{among treated--control pairs tied on all higher-priority layers, what is the
win--loss imbalance at stage $k$?}
PSNB uses these stage-conditional pairwise effects as its primitive layer-specific
estimands.

\section{Priority-Standardized Net Benefit}

\subsection{Definition}

Let $\bm{\alpha}=(\alpha_1,\ldots,\alpha_K)$ denote a prespecified
\emph{priority/credibility charter}, fixed before unblinding in the protocol or
statistical analysis plan, with $\alpha_k\ge 0$ and $\sum_{k=1}^K\alpha_k=1$.
The charter is part of the estimand definition, not a post hoc tuning parameter.

\begin{definition}[Priority-Standardized Net Benefit]
The \emph{Priority-Standardized Net Benefit} (PSNB) is
\begin{equation}\label{eq:psnb}
\Delta_{\mathrm{PS}}(\bm{\alpha})
=
\sum_{k=1}^K \alpha_k \Delta_k
=
\sum_{k=1}^K \alpha_k
\E\!\left[c_k(Y_1,Y_0)\mid R_k(Y_1,Y_0)=1\right].
\end{equation}
\end{definition}

Thus PSNB leaves the within-layer pairwise comparison rule unchanged and modifies
only how stage-specific information is recombined across the hierarchy. Because
each stage-conditional net benefit satisfies $\Delta_k\in[-1,1]$, PSNB itself is
bounded on the interpretable scale $[-1,1]$.

\paragraph{Interpretation.}
PSNB may be viewed through a two-stage thought experiment. First choose a layer
index $K^\star\in\{1,\ldots,K\}$ according to the charter probabilities
$\bm{\alpha}$. Then evaluate the expected treated--control comparison at that
layer among pairs that reach it. PSNB is the resulting weighted average of
stage-conditional win--loss imbalances. This is an interpretive re-expression,
not a sampling design; the estimand is defined directly by
equation~\eqref{eq:psnb}.

\subsection{What PSNB is not}

A natural alternative is a fixed weighted win--loss summary with pairwise score
\[
\psi_{\bm{\beta}}(y_1,y_0)
=
\sum_{k=1}^K \beta_k R_k(y_1,y_0)c_k(y_1,y_0),
\]
for some fixed weights $\bm{\beta}=(\beta_1,\ldots,\beta_K)$.
This is not the same object as PSNB.

\begin{proposition}[Fixed weighted win--loss remains reach-weighted]
\label{prop:weighted_reach}
Let
$\Delta_{\bm{\beta}}=\E\{\psi_{\bm{\beta}}(Y_1,Y_0)\}$.
Then
\[
\Delta_{\bm{\beta}}
=
\sum_{k=1}^K \beta_k r_k\Delta_k.
\]
Therefore, for any fixed $\bm{\beta}$, the influence of stage $k$ still depends on
the reach probability $r_k$.
\end{proposition}

\begin{proof}
Because $R_k\in\{0,1\}$,
\[
\E\{R_kc_k\}
=
\Prob(R_k=1)\E\{c_k\mid R_k=1\}
=
r_k\Delta_k.
\]
The result follows by linearity.
\end{proof}

Proposition~\ref{prop:weighted_reach} is why PSNB should be viewed as an
estimand-level change in \emph{across-layer aggregation}. Fixed weighted
win--loss summaries can attenuate or amplify a stage contribution, but they do
not decouple that contribution from reach. Reproducing PSNB from an unconditional
weighted score would require data-dependent weights
$\beta_k\propto \alpha_k/r_k$, which defeats the goal of a fixed, protocol-level
charter.

\subsection{Operational properties}

The following properties make PSNB useful as a prospective design and reporting
tool rather than only as a mathematical alternative to standard win summaries.

\begin{corollary}[Reach invariance]
\label{cor:reach_invariance}
Suppose two data-generating mechanisms share the same stage-conditional effects
$(\Delta_1,\ldots,\Delta_K)$ but have different reach probabilities
$(r_1,\ldots,r_K)$. Then PSNB is identical under the two mechanisms, whereas the
standard net benefit $\Delta=\sum_k r_k\Delta_k$ generally changes.
\end{corollary}

Corollary~\ref{cor:reach_invariance} formalizes the main substantive distinction
between PSNB and standard reach-weighted summaries. In practical terms, two trials
of the same treatment with identical stage-conditional effects but different event
rates, censoring patterns, or upstream tie rules will have the same PSNB under the
same charter, but will generally have different standard net benefits.

This invariance is a feature only when it matches the scientific question. Event
frequency is often clinically meaningful, and a reach-weighted win statistic may be
the more appropriate summary when the estimand is intended to reflect the realized
population burden of earlier events. PSNB instead targets stage-conditional effects
combined under layer weights set in advance. It should therefore be used when the
trial objective is to preserve the hierarchy while not letting the observed reach
distribution determine across-layer importance by default.

\begin{corollary}[Layer influence cap]
\label{cor:cap}
For every stage $k$, $|\Delta_k|\le 1$ and therefore
\[
|\alpha_k\Delta_k|\le \alpha_k.
\]
Hence the absolute contribution of stage $k$ to PSNB is bounded above by its
charter weight.
\end{corollary}

Corollary~\ref{cor:cap} provides a direct trial-design interpretation of the
charter. If a later layer receives weight $\alpha_K=0.10$, then its contribution
to the primary estimand can never exceed $0.10$ in absolute value, regardless of
how frequently it is reached. This is a worst-case bound rather than a typical
contribution: in realistic trials $|\Delta_k|$ is usually well below 1, so the
effective contribution is often much smaller. The charter weight therefore acts as
a hard ceiling, not a typical contribution.

\begin{proposition}[Tipping point for a later-layer weight]
\label{prop:tipping}
Fix a reference charter over the first $K-1$ layers,
$\tilde{\bm{\alpha}}=(\tilde\alpha_1,\ldots,\tilde\alpha_{K-1})$ with
$\tilde\alpha_j\ge 0$ and $\sum_{j=1}^{K-1}\tilde\alpha_j=1$.
Consider the one-parameter family of charters that holds the \emph{relative}
weighting among the first $K-1$ layers fixed at $\tilde{\bm{\alpha}}$ while
varying only the share assigned to the final layer:
\[
\alpha_j(\lambda)=(1-\lambda)\tilde\alpha_j,\quad j<K,
\qquad
\alpha_K(\lambda)=\lambda,
\qquad
\lambda\in[0,1].
\]
Let
\[
A=\sum_{j=1}^{K-1}\tilde\alpha_j\Delta_j.
\]
Then
\[
\Delta_{\mathrm{PS}}\{\bm{\alpha}(\lambda)\}
=
(1-\lambda)A+\lambda\Delta_K.
\]
If $\Delta_K\neq A$, the unique sign-change point is
\[
\lambda^\star=\frac{A}{A-\Delta_K}.
\]
When $\lambda^\star\in[0,1]$, the overall sign of PSNB changes exactly at
$\lambda=\lambda^\star$.
\end{proposition}

Proposition~\ref{prop:tipping} yields an interpretable protocol-level sensitivity
analysis: \emph{how much of the total charter may be assigned to the final layer
before the overall composite changes sign?} The family above is a natural
one-parameter construction because it varies only the total share allocated to the
final layer while keeping the relative weighting among the other layers fixed.
Other parameterizations are possible, but this one is simple to prespecify and
easy to communicate in a statistical analysis plan.

\begin{figure}[ht]
\centering
\includegraphics[width=0.85\linewidth]{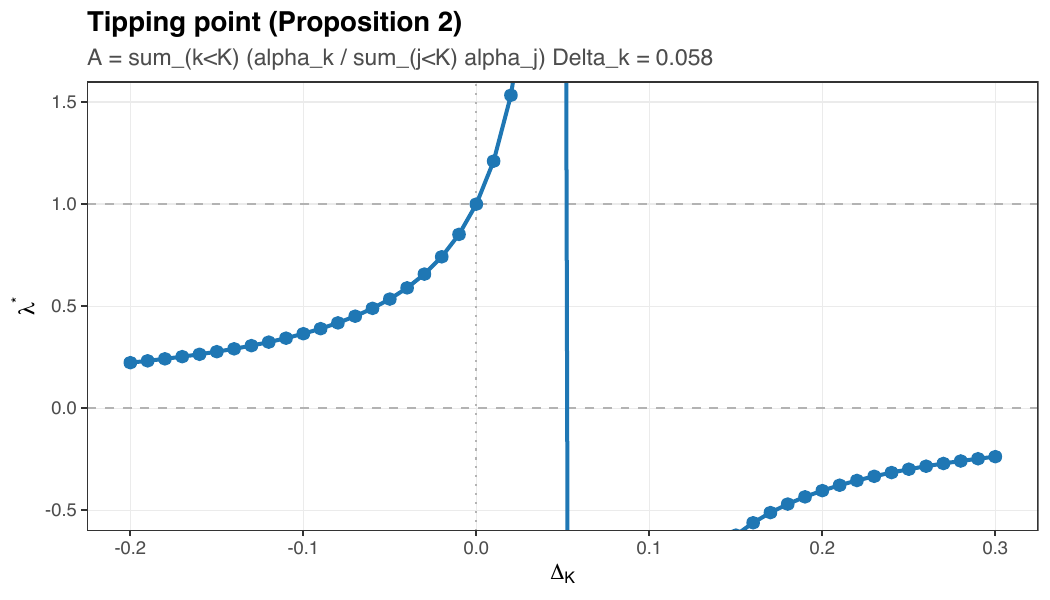}
\caption{Illustration of Proposition~\ref{prop:tipping}. The tipping point
$\lambda^\star = A/(A - \Delta_K)$ is plotted as a function of the final-layer
effect $\Delta_K$, with $\Delta_1$, $\Delta_2$ and the within-$\{1,2\}$ weights
fixed. Values of $\lambda^\star$ outside $[0,1]$ correspond to settings in which
\emph{no} admissible charter in the displayed family can flip the sign of the
composite; values inside $[0,1]$ pinpoint the smallest amount of total charter
that can be assigned to Layer~$K$ before the overall sign of PSNB changes.}
\label{fig:tipping_point}
\end{figure}

\begin{proposition}[Charter envelope]
\label{prop:envelope}
Let $\mathcal{A}\subset\R^K$ be any convex set of admissible charters defined by
linear constraints (for example, nonnegativity, sum-to-one, upper bounds, or
monotone priority constraints such as
$\alpha_1\ge\alpha_2\ge\cdots\ge\alpha_K$).
Then
\[
\underline\Delta_{\mathrm{PS}}
=
\min_{\bm{\alpha}\in\mathcal{A}}\Delta_{\mathrm{PS}}(\bm{\alpha}),
\qquad
\overline\Delta_{\mathrm{PS}}
=
\max_{\bm{\alpha}\in\mathcal{A}}\Delta_{\mathrm{PS}}(\bm{\alpha}),
\]
are attained at extreme points of $\mathcal{A}$ and can be computed by linear
programming.
\end{proposition}

Proposition~\ref{prop:envelope} suggests a protocol-friendly sensitivity analysis:
rather than selecting only one or two ad hoc alternative charters, one may
prespecify an admissible region and report the full PSNB range over that region.

\paragraph{Illustrative three-layer envelope.}
Suppose a three-layer trial prespecifies the admissible set
\[
\mathcal{A}
=
\left\{
\bm{\alpha}:
\alpha_k\ge 0,\ \sum_{k=1}^3\alpha_k=1,\ 
\alpha_3\le 0.20,\ 
\alpha_1\ge\alpha_2\ge\alpha_3
\right\},
\]
and suppose the estimated stage-conditional effects are
$(\Delta_1,\Delta_2,\Delta_3)=(0.08,0.02,0.20)$.
Then the charter envelope is
\[
\underline\Delta_{\mathrm{PS}}=0.05,
\qquad
\overline\Delta_{\mathrm{PS}}=0.092,
\]
attained at the admissible charters $(0.50,0.50,0)$ and $(0.60,0.20,0.20)$,
respectively. This converts a qualitative design statement such as
``the final layer should receive no more than 20\% of the total weight and later
layers should not outweigh earlier ones'' into a concrete interval for the primary
summary.

\begin{figure}[ht]
\centering
\includegraphics[width=0.85\linewidth]{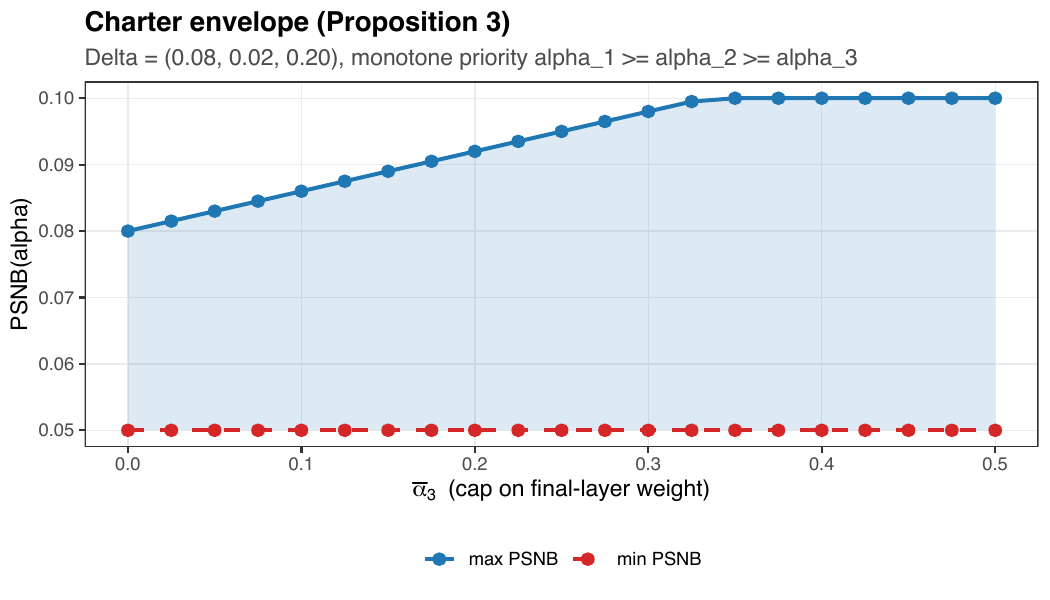}
\caption{Illustration of Proposition~\ref{prop:envelope}. The shaded band shows
$[\underline\Delta_{\mathrm{PS}}(\bar\alpha_3),\,\overline\Delta_{\mathrm{PS}}(\bar\alpha_3)]$
as the cap $\bar\alpha_3$ on the final-layer charter weight is swept from $0$
to $0.5$, with stage-conditional effects fixed at the illustrative
$(\Delta_1, \Delta_2, \Delta_3) = (0.08, 0.02, 0.20)$ and the admissible set
$\mathcal{A}$ defined by nonnegativity, sum-to-one, the cap $\alpha_3 \le \bar\alpha_3$,
and the monotone-priority constraint $\alpha_1 \ge \alpha_2 \ge \alpha_3$. The
width of the band quantifies how much disagreement remains across charters once
the stage-conditional effects are fixed; tightening $\bar\alpha_3$ narrows the band.}
\label{fig:charter_envelope}
\end{figure}

\subsection{Choosing and documenting the charter}

The charter $\bm{\alpha}$ should be chosen before unblinding, documented in the
protocol or statistical analysis plan, and locked before treatment-arm outcome
comparisons are reviewed. The charter is part of the estimand, not a tuning
parameter to be optimized after seeing efficacy results. Its justification should
address two considerations:
\begin{enumerate}[label=(\roman*)]
\item \textbf{Clinical priority:} how important is each layer in relation to the
trial objective?
\item \textbf{Measurement credibility:} how resistant is each layer to bias,
missingness, or contextual effects under the planned trial design?
\end{enumerate}

We do not prescribe universal numerical values for these weights. Charter
selection is a substantive judgment for the trial team, but it should be made
through a reproducible process. In a confirmatory trial, the process should specify
who has authority to recommend and approve the charter; what evidence is used; how
patient, clinician, statistical, and regulatory perspectives are represented; and
how disagreements are resolved before database lock. A patient-reported or
health-status endpoint may appropriately receive substantial weight when it is well
validated, prospectively defined, expected to be well observed, and collected under
conditions that support credible interpretation. Conversely, a layer that is highly
susceptible to expectation bias, differential missingness, or post-randomization
measurement artifacts should have that vulnerability reflected in its charter
weight or in the sensitivity set.

A practical workflow is:
\begin{enumerate}[leftmargin=*]
\item define the clinical role of each layer and choose baseline
clinical-priority weights;
\item identify evidence for each layer's measurement credibility, including
validation evidence, expected completeness, blinding, ascertainment procedures,
and prior trial or external-control information relevant to contextual bias;
\item apply credibility modifiers to layers judged more or less vulnerable to
bias or missingness, and renormalize to sum to one;
\item document the resulting primary charter, the rationale for every nonzero
weight, and any constraints such as upper bounds or monotone-priority rules;
\item prespecify either a small sensitivity menu or a charter envelope, together
with the inferential status of those analyses.
\end{enumerate}

The documentation should include enough information for an external reviewer to
reconstruct the decision: the hierarchy, the candidate weight set, the evidence
supporting any credibility adjustment, the resulting primary charter, the
admissible sensitivity set, and the operating characteristics used to assess the
selected design. If more than one charter is intended to support confirmatory
claims, the multiplicity strategy should be prespecified. Otherwise, sensitivity
charters and charter envelopes should be interpreted as robustness analyses rather
than additional primary tests.

\paragraph{Illustrative charter construction.}
Consider a three-layer hierarchy with baseline clinical-priority weights
$(0.50,0.30,0.20)$. If the third layer is judged more vulnerable to bias or
missingness than the first two, one might apply credibility modifiers
$(1.0,1.0,0.5)$, yielding unnormalized weights $(0.50,0.30,0.10)$. After
renormalization, the resulting charter is
\[
(0.556,\,0.333,\,0.111),
\]
which may be rounded for reporting to $(0.56,0.33,0.11)$. This example is purely
illustrative, but it shows how the charter can be constructed transparently from
a combination of clinical priority and measurement credibility.

\subsection{Priority-Standardized Win Ratio}

For readers who prefer a ratio-scale summary, define the
\emph{Priority-Standardized Win Ratio} (PSWR) by
\begin{equation}\label{eq:pswr}
\mathrm{WR}_{\mathrm{PS}}(\bm{\alpha})
=
\frac{\sum_{k=1}^K \alpha_k w_k}
     {\sum_{k=1}^K \alpha_k \ell_k},
\end{equation}
whenever the denominator is positive.

PSWR is a companion rather than a separate primary contribution. Like PSNB, it
standardizes across-layer aggregation using the charter rather than reach.
Writing
\[
\bar w=\sum_{k=1}^K\alpha_k w_k,
\qquad
\bar \ell=\sum_{k=1}^K\alpha_k \ell_k,
\]
we have
\[
\Delta_{\mathrm{PS}}=\bar w-\bar\ell,
\qquad
\mathrm{WR}_{\mathrm{PS}}=\bar w/\bar\ell.
\]
Therefore $\mathrm{WR}_{\mathrm{PS}}>1$ if and only if
$\Delta_{\mathrm{PS}}>0$.
This directional concordance does not mean that PSWR is numerically comparable to
the standard win ratio: the two summaries live on different aggregation rules.
PSWR should therefore be interpreted within a given charter, not compared
numerically with a standard WR from another trial or even with PSWR under a
different charter.

\section{Estimation and large-sample inference}

\subsection{Observed-data estimator}

Let $(Z_i,Y_i)$, $i=1,\ldots,n$, denote trial observations, with $n_1$ treated
participants and $n_0$ control participants. For each treated--control pair,
compute the prespecified stage comparison scores $c_k(Y_i,Y_j)$ and reach indicators
$R_k(Y_i,Y_j)$.

For each stage $k$, define the empirical stage contribution and reach probability by
\[
\widehat U_k
=
\frac{1}{n_1n_0}
\sum_{i:Z_i=1}\sum_{j:Z_j=0} R_k(Y_i,Y_j)c_k(Y_i,Y_j),
\qquad
\widehat r_k
=
\frac{1}{n_1n_0}
\sum_{i:Z_i=1}\sum_{j:Z_j=0} R_k(Y_i,Y_j).
\]
The stage-conditional net benefit is then estimated by
\[
\widehat\Delta_k=\frac{\widehat U_k}{\widehat r_k},
\]
and the PSNB estimator is
\[
\widehat\Delta_{\mathrm{PS}}(\bm{\alpha})
=
\sum_{k=1}^K \alpha_k\widehat\Delta_k,
\]
whenever $\widehat r_k>0$ for all stages receiving nonzero weight.

This estimator has computational complexity $O(n_1n_0K)$.
For larger trials, Monte Carlo pair subsampling or structure-exploiting
implementations may be used, as is common in generalized pairwise comparison analyses.

\paragraph{Censoring.}
The estimator is defined on the observed-data comparison functions $c_k$ chosen for
each layer. For time-to-event layers subject to administrative censoring, a practical
primary implementation is to prespecify a common horizon $\tau$ and a restricted-horizon
ordering rule: an event-free participant through $\tau$ beats a participant with the
event before $\tau$; among participants with events before $\tau$, the later event
time wins; and pairs that cannot be ordered because of censoring before $\tau$ are
handled by the prespecified rule, for example as ties. This produces a well-defined
observed-data comparison for administrative censoring and makes the horizon visible
in the estimand.

When nonadministrative or informative censoring is expected, the analysis plan should
not rely on PSNB to solve that problem. It should instead prespecify a censoring
estimand and analysis strategy, such as inverse probability of censoring weighting,
principal-stratum or hypothetical strategies, or sensitivity analyses for censoring
assumptions. PSNB changes the across-layer aggregation rule after the stage-specific
comparison rules have been defined; it is complementary to, not a replacement for,
methods that define or estimate the censored pairwise target \citep{Mao2024EstimandWR}.

\subsection{Assumptions}

\begin{assumption}[Consistency, no interference, and randomization]
\label{ass:rand}
Observed outcomes satisfy $Y=Y(Z)$, there is no interference between participants,
and randomization identifies the arm-specific outcome distributions.
\end{assumption}

\begin{assumption}[Independent sampling within arm]
\label{ass:iid}
Within each arm, observations are i.i.d.\ draws from the corresponding arm-specific
distribution, and the treated and control samples are independent.
\end{assumption}

\begin{assumption}[Positive reach for weighted stages]
\label{ass:reach}
For every stage $k$ with $\alpha_k>0$, the reach probability
$r_k=\Prob(R_k=1)$ is positive. For stable ratio inference, we assume that $r_k$
is bounded away from zero for all stages with nonzero weight.
\end{assumption}

Assumption~\ref{ass:iid} is the same working sampling structure used in standard
large-sample win-ratio inference \citep{BebuLachin2016}. In stratified randomized
trials, PSNB may be estimated within strata and combined across strata; clustered
designs would require cluster-robust extensions, which we do not pursue here.

Because $R_k\in\{0,1\}$ and $c_k\in\{-1,0,+1\}$, all kernels used below are bounded.

\subsection{Variance estimation and practical inference}

Define the stage-$k$ kernels
\[
h_k(y_1,y_0)=R_k(y_1,y_0)c_k(y_1,y_0),
\qquad
g_k(y_1,y_0)=R_k(y_1,y_0).
\]
A practical projection-based variance estimator may be computed directly from the data.

For each treated observation $i$ and stage $k$, define
\[
\widehat\varphi_{1,k}(Y_i)
=
\frac{1}{n_0}\sum_{j:Z_j=0}
\frac{h_k(Y_i,Y_j)-\widehat\Delta_k\,g_k(Y_i,Y_j)}{\widehat r_k},
\]
and for each control observation $j$ define
\[
\widehat\varphi_{0,k}(Y_j)
=
\frac{1}{n_1}\sum_{i:Z_i=1}
\frac{h_k(Y_i,Y_j)-\widehat\Delta_k\,g_k(Y_i,Y_j)}{\widehat r_k}.
\]
Aggregate these stagewise contributions using the charter:
\[
\widehat\psi_1(Y_i)=\sum_{k=1}^K \alpha_k\widehat\varphi_{1,k}(Y_i),
\qquad
\widehat\psi_0(Y_j)=\sum_{k=1}^K \alpha_k\widehat\varphi_{0,k}(Y_j).
\]
The plug-in projection variance estimator is
\begin{equation}\label{eq:plugin_var}
\widehat\sigma^2_{\mathrm{PS}}
=
\frac{n}{n_1}\widehat{\mathrm{Var}}_1\!\left[\widehat\psi_1(Y_i)\right]
+
\frac{n}{n_0}\widehat{\mathrm{Var}}_0\!\left[\widehat\psi_0(Y_j)\right],
\end{equation}
where $\widehat{\mathrm{Var}}_1$ and $\widehat{\mathrm{Var}}_0$ denote sample variances
within the treated and control arms, respectively. The resulting standard error is
\[
\widehat{\mathrm{se}}\!\left\{\widehat\Delta_{\mathrm{PS}}(\bm{\alpha})\right\}
=
\widehat\sigma_{\mathrm{PS}}/\sqrt{n}.
\]

A Wald-type $95\%$ confidence interval is then
\[
\widehat\Delta_{\mathrm{PS}}(\bm{\alpha})
\ \pm\
1.96\,
\widehat{\mathrm{se}}\!\left\{\widehat\Delta_{\mathrm{PS}}(\bm{\alpha})\right\}.
\]

We recommend reporting both:
\begin{enumerate}[leftmargin=*]
\item a projection-based Wald interval using \eqref{eq:plugin_var}, and
\item a nonparametric bootstrap interval stratified by treatment arm, resampling
participants (not pairs) within each arm.
\end{enumerate}
Because each participant contributes to many treated--control pairs, the bootstrap
should resample participants within arm rather than resampling pairwise comparisons.
In practice, at least 1{,}000 bootstrap replicates is a reasonable default.

The formulas above are directly implementable in standard statistical software, and
the appendix summarizes the participant-level computation.
Section~\ref{sec:sim_results} shows that the large-sample Wald approximation is well
calibrated in the simulation study even at $n_1=n_0=150$ per arm.

\subsection{Large-sample justification}

Let
\[
U_k=\E\{R_k(Y_1,Y_0)c_k(Y_1,Y_0)\},
\qquad
r_k=\E\{R_k(Y_1,Y_0)\},
\qquad
\Delta_k=U_k/r_k.
\]
Then $\widehat U_k$ and $\widehat r_k$ are two-sample $U$-statistics.

\begin{proposition}[Consistency and rate]
\label{prop:consistency}
Under Assumptions~\ref{ass:rand}--\ref{ass:reach},
$\widehat U_k\xrightarrow{p}U_k$ and $\widehat r_k\xrightarrow{p}r_k$ for each
weighted stage $k$, and therefore
\[
\widehat\Delta_{\mathrm{PS}}(\bm{\alpha})
\xrightarrow{p}
\Delta_{\mathrm{PS}}(\bm{\alpha}).
\]
Moreover,
\[
\widehat\Delta_{\mathrm{PS}}(\bm{\alpha})
-
\Delta_{\mathrm{PS}}(\bm{\alpha})
=
O_p(n^{-1/2}).
\]
\end{proposition}

The asymptotic normality of PSNB follows from standard two-sample $U$-statistic
projection arguments combined with a delta-method expansion of the stagewise ratio
estimators $\widehat\Delta_k=\widehat U_k/\widehat r_k$.

\begin{theorem}[Asymptotic normality]
\label{thm:clt}
Let $p_n=n_1/n\to p\in(0,1)$ as $n\to\infty$. Under
Assumptions~\ref{ass:rand}--\ref{ass:reach},
\[
\sqrt{n}\left(
\widehat\Delta_{\mathrm{PS}}(\bm{\alpha})
-
\Delta_{\mathrm{PS}}(\bm{\alpha})
\right)
\xrightarrow{d}
N\!\left(0,\sigma^2_{\mathrm{PS}}\right),
\]
where
\[
\sigma^2_{\mathrm{PS}}
=
\frac{1}{p}\Var\{\psi_1(Y_1)\}
+
\frac{1}{1-p}\Var\{\psi_0(Y_0)\},
\]
and
\[
\psi_1(y)
=
\sum_{k=1}^K \alpha_k
\E\!\left[
\frac{h_k(y,Y_0)-\Delta_k g_k(y,Y_0)}{r_k}
\right],
\qquad
\psi_0(y)
=
\sum_{k=1}^K \alpha_k
\E\!\left[
\frac{h_k(Y_1,y)-\Delta_k g_k(Y_1,y)}{r_k}
\right].
\]
\end{theorem}

We deliberately describe this as an \emph{influence-function-based asymptotic linear
representation}. That is the object needed for variance estimation, Wald inference,
bootstrap calibration, and future covariate-adjusted extensions. The displayed
projection terms are the quantities used in the plug-in variance estimator above.

\paragraph{PSWR.}
For the ratio-scale companion summary $\mathrm{WR}_{\mathrm{PS}}$, inference is most
naturally performed on the log scale. Specifically,
\[
\sqrt{n}\left(
\log\widehat{\mathrm{WR}}_{\mathrm{PS}}
-
\log\mathrm{WR}_{\mathrm{PS}}
\right)
\]
is asymptotically normal by the delta method applied to the charter-weighted win and
loss probabilities. Because PSWR is secondary in this paper, confirmatory inference
and operating-characteristic claims are based on PSNB; PSWR is reported as a
ratio-scale interpretive companion.

\paragraph{Scope.}
Closed-form sample size formulae for PSNB, analogous to those available for standard
win-ratio analyses \citep{Mao2022SampleSize}, are not developed here. The simulation
study instead focuses on operating characteristics, simulation-based design power, and
the efficiency--robustness trade-off induced by the charter; a closed-form design
treatment is a topic for future work.

\section{Simulation study}
\label{sec:sim}

\subsection{Aims}

The simulation study is organized around the properties that most clearly distinguish PSNB
from standard hierarchical summaries. The experiments address seven practical questions:
\begin{enumerate}[leftmargin=*]
\item Does projection-based Wald inference for PSNB maintain nominal type I error?
\item If stage-conditional effects are held fixed while reach changes, does PSNB remain
stable while standard win summaries move?
\item Do threshold or tie-rule changes upstream mechanically shift a standard composite
toward the last layer, even when the last-layer effect is unchanged?
\item How quickly does spurious rejection arise when a dominant last layer is subject to
open-label bias?
\item How sensitive is a hierarchical composite to differential missingness in the
last layer, and how does that sensitivity change as the last-layer weight is reduced?
\item What is the power cost of robustness when the last-layer weight is capped?
\item Can a prespecified bias budget translate an expected magnitude of late-layer bias
into a maximum last-layer charter weight, and does the resulting weight keep spurious
rejection near nominal in simulation?
\end{enumerate}

\subsection{Data-generating mechanism}
\label{sec:dgm}

Unless otherwise stated, simulations use a three-layer hierarchy chosen to resemble a broad
class of contemporary clinical-trial composites:
\begin{enumerate}[label=(\roman*)]
\item a hard time-to-event component through a fixed horizon $\tau$,
\item an intermediate count or recurrent-event component through $\tau$, and
\item a late health-status or PRO responder component at $\tau$.
\end{enumerate}

Participants are generated independently within arm. To induce realistic within-participant
dependence across layers, each participant has an unobserved severity variable
$S\sim N(0,1)$ such that worse severity increases event risk, increases count burden,
and reduces the probability of late response. Treatment effects act through layer-specific
parameters and may differ by scenario.

\paragraph{Layer 1: hard event.}
Let $T_z$ denote a potential event time under arm $z\in\{0,1\}$.
Conditional on severity, $T_z$ follows an exponential model with hazard
$\lambda_z(S)=\lambda_{0z}\exp(\eta_1 S)$.
Comparisons are truncated at $\tau$: event-free through $\tau$ beats an event before $\tau$,
and later event time is better when both have events.

\paragraph{Layer 2: intermediate event burden.}
Let $H_z$ denote a count outcome through $\tau$ with
$H_z\mid S\sim \mathrm{Poisson}\{\mu_z\exp(\eta_2 S)\}$.
Comparisons favor fewer events; a tie margin $m_H\ge 0$ may be introduced so that
differences within $\pm m_H$ are treated as ties.

\paragraph{Layer 3: late health-status or PRO layer.}
Let $K_z$ denote a continuous change score with
$K_z\mid S\sim N(\mu_{K,z}-\eta_3 S,\sigma_K^2)$.
The layer is analyzed either as a responder indicator
$\Ind\{K_z\ge \delta_{\mathrm{resp}}\}$ or via a margin rule on pairwise differences.
Open-label bias is introduced by adding a treated-arm shift $b$ to the observed late-layer
score. Differential missingness is introduced by a logistic model for the probability that
the late-layer score is observed.

\paragraph{Charters.}
Across experiments we use two baseline three-layer charters:
\[
\bm{\alpha}^{(A)}=(0.50,0.30,0.20),
\qquad
\bm{\alpha}^{(B)}=(0.60,0.30,0.10).
\]
These are not recommended defaults. They serve only to illustrate moderate versus more
conservative weighting of the final layer.

For the open-label bias calibration in Experiment~D, we also use two
bias-budget charters. Suppose the late patient-reported layer is compared using a
continuous margin $\delta$, has within-arm standard deviation $\sigma$, and could
be affected by an additive treated--control pairwise bias $b$. Under a normal
working approximation for the pairwise difference, the bias-induced apparent
late-layer net benefit is
\[
d_3(b)
=
\Phi\left\{\frac{b-\delta}{\sigma\sqrt{2}}\right\}
-
\Phi\left\{\frac{-b-\delta}{\sigma\sqrt{2}}\right\}.
\]
If the design tolerance for the maximum estimand-scale contribution of this bias is
$\tau$, then a transparent cap is
\[
\alpha_3 \le \min\{1,\tau/d_3(b)\}
\quad\text{when } d_3(b)>0.
\]
The treatment-only charter $\bm{\alpha}^{(C)}=(0.57,0.38,0.05)$ uses a Layer-3 cap
of $5\%$, consistent with $\delta=5$, $\sigma=10$, $b=5$, and a bias budget of about
$0.013$ on the PSNB scale. The compounded-bias charter
$\bm{\alpha}^{(D)}=(0.588,0.392,0.02)$ uses a Layer-3 cap of $2\%$, consistent with
$\delta=5$, $\sigma=10$, $b=10$, and a bias budget of about $0.010$. These numerical
values are design examples, not universal recommendations; in an actual trial,
$b$, $\delta$, $\sigma$, and $\tau$ should be justified from prespecified evidence
and documented with the charter. In both charters, the Layer~1:Layer~2 ratio is kept
at $60{:}40$ after fixing the Layer-3 cap.

\subsection{Estimands and comparators}

Each simulated trial reports:
\begin{enumerate}[leftmargin=*]
\item standard reach-weighted net benefit $\Delta$,
\item standard win ratio $\mathrm{WR}$,
\item a fixed weighted win--loss summary
$\Delta_{\bm{\beta}}=\sum_k\beta_k r_k\Delta_k$ with
$\bm{\beta}=(0.50,0.30,0.20)$ for direct comparison,
\item PSNB under charters $\bm{\alpha}^{(A)}$ and $\bm{\alpha}^{(B)}$, with
$\bm{\alpha}^{(C)}$ and $\bm{\alpha}^{(D)}$ added in the open-label bias stress test,
\item PSWR under the same primary charters,
\item the last-layer reach probability $r_K$, and
\item the last-layer share $U_K/\Delta$ when $\Delta\neq 0$.
\end{enumerate}

Empirical rejection probabilities are based on two-sided large-sample Wald tests at
nominal level $0.05$ using participant-level projection variances. For the standard
WR comparator, the tested null is no treated--control win imbalance,
$W=L$ (equivalently $\mathrm{WR}=1$ when both win and loss probabilities are
positive), implemented on the net-benefit scale $W-L=0$ for the rejection-rate
comparisons. For PSNB, the tested null is
$\Delta_{\mathrm{PS}}(\bm{\alpha})=0$ under the prespecified charter. Bootstrap
calculations in Experiment~A are used as an interval-width check, not as the primary
testing method.

\subsection{Experiment A: Global-null calibration}

Under the global null, treatment has no effect on any layer. We evaluate
empirical type I error for two-sided Wald tests at nominal level $0.05$ over a
sample-size sweep $n_1=n_0\in\{75,150,300,600\}$, with $2{,}000$ Monte Carlo
replicates at each cell for $n\le 300$ and $1{,}000$ replicates at $n=600$.
In addition, we run a bootstrap validation block at $n_1=n_0=150$ ($100$
replicates, $B=300$ stratified nonparametric bootstrap resamples per
replicate) comparing the median half-width of the projection-based Wald
confidence interval implied by Theorem~\ref{thm:clt} to the median half-width
of the bootstrap interval, for both charters $\bm{\alpha}^{(A)}$ and
$\bm{\alpha}^{(B)}$.

\subsection{Experiment B: Reach variation with fixed stage-conditional effects}

This experiment is designed specifically to illustrate
Corollary~\ref{cor:reach_invariance}. Parameters are numerically calibrated so that the
stage-conditional effects $(\Delta_1,\Delta_2,\Delta_3)$ are approximately stable across
settings, while control-arm hard-event rarity and upper-layer tie frequency are varied to
change the reach distribution. We then compare how standard net benefit, fixed weighted
win--loss, and PSNB move across the sweep.

\subsection{Experiment C: Threshold- and tie-induced spillover}

To isolate the effect of upstream tie rules, we hold the treatment effect at the final
layer fixed and vary the tie margin in the intermediate layer. This increases or decreases
the flow of pairs to the final layer without changing the final-layer effect itself.
The goal is to quantify how much of a change in the composite summary is due to the altered
tie rule rather than to any change in clinical benefit.

\subsection{Experiment D: Open-label bias stress test}

Under a global null for the hard and intermediate layers and no true late-layer treatment
effect, we add a treated-arm late-layer bias shift
$b\in\{0,2.5,5,7.5,10\}$.
At $n_1=n_0=175$ and 1500 replicates per setting, we estimate the rejection probability for
standard WR and for PSNB under the two baseline charters and the two bias-budget
charters. The shift $b$ is interpreted as the induced treated--control pairwise
separation on the late layer; for the rejection-rate calculation, only this pairwise
separation matters, not whether it arises from treatment expectation, control
disappointment, or both.

\subsection{Experiment E: Differential missingness stress test}

We generate late-layer missingness from a logistic model depending on arm, severity, and the
latent late-layer potential outcome. Missingness rates are chosen so that the final layer is
frequently reached but variably observed. The primary analysis treats missing late-layer values
as ties, thereby isolating the across-layer aggregation question from any additional imputation
model. The goal is comparative rather than absolute: how strongly do standard hierarchical
summaries versus PSNB respond when the dominant last layer becomes differentially incomplete?

\subsection{Experiment F: Efficiency--robustness frontier and design power}

Under alternative scenarios with genuine treatment benefit, we first vary the maximum
allowable last-layer weight $\alpha_3\le\bar\alpha_3$ over a prespecified grid and
estimate empirical power at $n_1=n_0=175$. This frontier quantifies the cost of
robustness when the final layer is both highly informative and potentially less
credible. We then perform a simulation-based design-power sweep over
$n_1=n_0\in\{175,250,350,500,650,800\}$ for the standard Win Ratio and for the two
baseline PSNB charters. The broad-benefit profile uses a hard-event hazard ratio of
$0.80$, an intermediate-event rate ratio of $0.85$, and a Layer-3 mean shift of $2$
points. The final-layer-dominated profile has no effect on Layers~1--2 and a Layer-3
mean shift of $4$ points.

\subsection{Experiment G: Bias-budget calibration of the last-layer weight}

Experiment~D showed that the two bias-budget charters, $\bm{\alpha}^{(C)}$ and
$\bm{\alpha}^{(D)}$, control spurious rejection at their target bias benchmarks.
Experiment~G makes the underlying construction explicit and examines it directly.
Recall the bias-budget expression used in Section~\ref{sec:dgm} to define those
charters: under a normal working model for the late-layer pairwise difference with
margin $\delta$ and within-arm standard deviation $\sigma$, an additive
treated--control bias $b$ induces an apparent late-layer net benefit
$d_3(b)=\Phi\{(b-\delta)/(\sigma\sqrt2)\}-\Phi\{(-b-\delta)/(\sigma\sqrt2)\}$, and a
tolerance $\tau$ on the estimand-scale contribution of that bias yields the cap
$\bar\alpha_3=\min\{1,\tau/d_3(b)\}$. We (i) trace $\bar\alpha_3$ as a function of the
expected bias $b$ for tolerances $\tau\in\{0.01,0.025,0.05,0.10\}$ with $\delta=5$ and
$\sigma=10$, and (ii) run a companion calibration simulation that sweeps the last-layer
weight $\alpha_3$ over a grid against bias levels
$b\in\{0,2.5,3.6,5,7.5,10\}$ at $n_1=n_0=175$ with $600$ replicates per cell, recording
the empirical rejection rate under the global null. The two together separate the
analytical design heuristic (i) from its operating-characteristic check (ii).

\section{Simulation results}
\label{sec:sim_results}

\subsection{Experiment A: Type I error}

Table~\ref{tab:type1} reports empirical type I error under the global null
across the full sample-size sweep. Every cell is within two Monte Carlo
standard errors of nominal $0.05$: rejection rates for the standard Win
Ratio lie in $[0.053, 0.057]$, PSNB$(\bm{\alpha}^{(A)})$ in $[0.048, 0.055]$,
and PSNB$(\bm{\alpha}^{(B)})$ in $[0.048, 0.054]$. Calibration holds down
to $n_1=n_0=75$ per arm without degradation, supporting the use of the
projection-based Wald approximation of Theorem~\ref{thm:clt} at realistic
trial sample sizes.

\begin{table}[ht]
\centering
\caption{Empirical type I error under the global null at two-sided $\alpha = 0.05$, across a sample-size sweep. Rejection rates are shown with Monte Carlo standard errors; all entries lie within two Monte Carlo SEs of nominal 0.05.}
\label{tab:type1}
\begin{tabular}{@{}lllll@{}}
\toprule
method & $n$ per arm & rejection & MC SE & $n_{\mathrm{rep}}$ \\
\midrule
WR & 75 & 0.053 & 0.005 & 2000 \\
PSNB(0.50,0.30,0.20) & 75 & 0.048 & 0.005 & 2000 \\
PSNB(0.60,0.30,0.10) & 75 & 0.049 & 0.005 & 2000 \\
WR & 150 & 0.057 & 0.005 & 2000 \\
PSNB(0.50,0.30,0.20) & 150 & 0.053 & 0.005 & 2000 \\
PSNB(0.60,0.30,0.10) & 150 & 0.048 & 0.005 & 2000 \\
WR & 300 & 0.055 & 0.005 & 2000 \\
PSNB(0.50,0.30,0.20) & 300 & 0.053 & 0.005 & 2000 \\
PSNB(0.60,0.30,0.10) & 300 & 0.054 & 0.005 & 2000 \\
WR & 600 & 0.055 & 0.007 & 1000 \\
PSNB(0.50,0.30,0.20) & 600 & 0.055 & 0.007 & 1000 \\
PSNB(0.60,0.30,0.10) & 600 & 0.051 & 0.007 & 1000 \\
\bottomrule
\end{tabular}
\end{table}

\paragraph{Bootstrap--Wald variance validation.}
Because Theorem~\ref{thm:clt} relies on a projection-based variance
estimator, we also compare the closed-form Wald half-width against a
stratified nonparametric bootstrap half-width under the global null.
Table~\ref{tab:bootstrap_vs_wald} reports the median of each quantity over
$100$ replicates at $n_1=n_0=150$ with $B=300$ bootstrap resamples per
replicate. The ratio of bootstrap to Wald half-widths is $0.98$ for charter
$\bm{\alpha}^{(A)}$ and $0.97$ for charter $\bm{\alpha}^{(B)}$: the
projection variance estimator agrees with the nonparametric bootstrap to
within Monte Carlo noise, consistent with the asymptotic linear
representation in Theorem~\ref{thm:clt}.

\begin{table}[ht]
\centering
\caption{Bootstrap--Wald variance validation at $n = 150$ per arm, 100 replicates, $B = 300$ bootstrap resamples under the global null. Each charter's projection-based Wald half-width is compared to the stratified nonparametric bootstrap half-width; ratios near 1.00 indicate that Theorem~1's projection variance estimator agrees with the bootstrap to within Monte Carlo noise.}
\label{tab:bootstrap_vs_wald}
\begin{tabular}{@{}lllll@{}}
\toprule
charter & median estimate & median Wald half & median boot half & boot / Wald \\
\midrule
A & -0.0040 & 0.0650 & 0.0637 & 0.9811 \\
B & -0.0030 & 0.0655 & 0.0633 & 0.9652 \\
\bottomrule
\end{tabular}
\end{table}

\subsection{Experiment B: Reach variation with fixed stage-conditional effects}

Figure~\ref{fig:reach_invariance} is the core identifying simulation for the
paper. A coordinate-wise bisection calibration is run on the data-generating
mechanism so that the stage-conditional effects $(\hat\Delta_1,
\hat\Delta_2, \hat\Delta_3)$ sit at the target $(0.05, 0.05, 0.08)$ within a
tolerance of $0.005$ in every scenario, while control-arm hard-event rarity
and upper-layer tie frequency vary across scenarios to change the last-layer
reach probability. The sweep covers $r_3$ values ranging from $0.079$ in the
low-reach scenario to $0.853$ in the high-reach scenario, an order of
magnitude. Across this range, the standard net benefit $\Delta$ moves by
$0.099$ and the fixed weighted win--loss summary $\Delta_{\bm{\beta}}$ moves
by $0.024$, tracking the reach distribution as predicted by
equation~\eqref{eq:standard_reach_weight} and
Proposition~\ref{prop:weighted_reach}. PSNB, by contrast, moves by at most
$0.006$ for both charters---roughly $16\times$ less than $\Delta$---which is
the numerical content of Corollary~\ref{cor:reach_invariance}.

\begin{figure}[ht]
\centering
\includegraphics[width=0.92\linewidth]{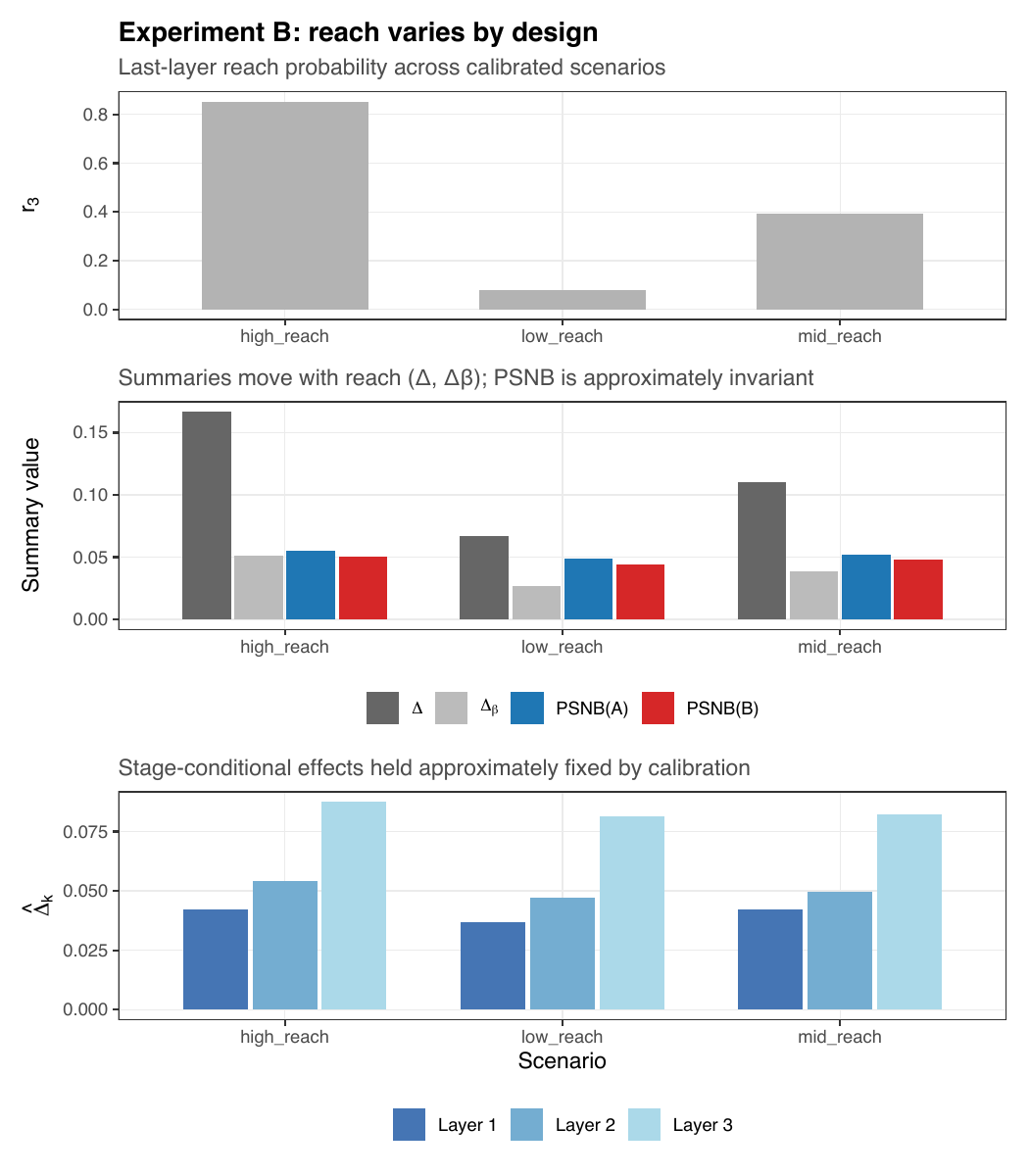}
\caption{Experiment B: standard net benefit $\Delta$ and fixed weighted
win--loss $\Delta_{\bm{\beta}}$ move with the last-layer reach probability
$r_3$, whereas PSNB remains approximately invariant when stage-conditional
effects are held fixed. Top panel: $r_3$ by scenario. Middle panel: $\Delta$,
$\Delta_{\bm{\beta}}$, and the two PSNB charters by scenario. Bottom panel:
stage-conditional effects $\hat\Delta_k$ confirming the calibration target is
met across scenarios.}
\label{fig:reach_invariance}
\end{figure}

Table~\ref{tab:reach_invariance} reports the numerical summary.

\begin{table}[ht]
\centering
\caption{Experiment B: reach-invariance summary. After calibration, the stage-conditional effects $(\hat\Delta_1, \hat\Delta_2, \hat\Delta_3)$ are held approximately at target across scenarios with markedly different last-layer reach $r_3$. Standard summaries $\Delta$ and $\Delta_\beta$ move with $r_3$; PSNB is approximately invariant (Corollary~1).}
\label{tab:reach_invariance}
\begin{tabular}{@{}llllllllll@{}}
\toprule
scenario & $r_3$ & $\hat\Delta_1$ & $\hat\Delta_2$ & $\hat\Delta_3$ & $\Delta$ & $\Delta_\beta$ & PSNB(A) & PSNB(B) & $n_{\mathrm{rep}}$ \\
\midrule
low\_reach & 0.079 & 0.037 & 0.047 & 0.081 & 0.067 & 0.027 & 0.049 & 0.044 & 1000 \\
mid\_reach & 0.391 & 0.042 & 0.050 & 0.082 & 0.110 & 0.038 & 0.052 & 0.048 & 1000 \\
high\_reach & 0.853 & 0.042 & 0.054 & 0.088 & 0.167 & 0.051 & 0.055 & 0.050 & 1000 \\
\bottomrule
\end{tabular}
\end{table}

\paragraph{Take-away.}
The results make visually obvious what
Proposition~\ref{prop:weighted_reach} shows algebraically: fixed weighted
win--loss attenuates or amplifies stages differently, but it does
\emph{not} remove the stochastic role of reach. PSNB does. The residual
PSNB variation of $\sim 0.006$ across scenarios reflects the finite
calibration tolerance of the data-generating mechanism and the fact that
the realised $\hat\Delta_1$ slightly undershoots the target of $0.05$ in
the lowest-reach scenario; it is not an indication that PSNB itself depends
on reach.

\subsection{Experiment C: Threshold- and tie-induced spillover}

Figure~\ref{fig:spillover} and Table~\ref{tab:spillover} show how changing an upstream
tie rule alters the overall composite even when the final-layer effect is unchanged.
Sweeping the intermediate-layer tie margin over $m_H\in\{0,1,\ldots,5\}$ drives the
final-layer reach from $r_3 = 0.142$ at $m_H=0$ all the way up to $r_3 = 0.720$ at
$m_H=5$. Over the same sweep the standard net benefit moves from $\Delta = 0.020$ to
$\Delta = 0.100$---a fivefold change---and the last-layer share $U_3/\Delta$ stays
essentially pinned at $1$ (range $[0.992, 1.038]$), i.e.\ virtually \emph{all} of
the standard summary is driven by Layer~3 once ties upstream are permitted. The fixed
weighted win--loss summary $\Delta_\beta$ tracks $\Delta$ in the same direction,
moving from $0.004$ to $0.020$. By contrast, PSNB is essentially flat: PSNB$(\bm{\alpha}^{(A)})$
moves by only $0.0018$ across the entire sweep and PSNB$(\bm{\alpha}^{(B)})$ moves by
at most $0.002$.

\begin{figure}[ht]
\centering
\includegraphics[width=0.92\linewidth]{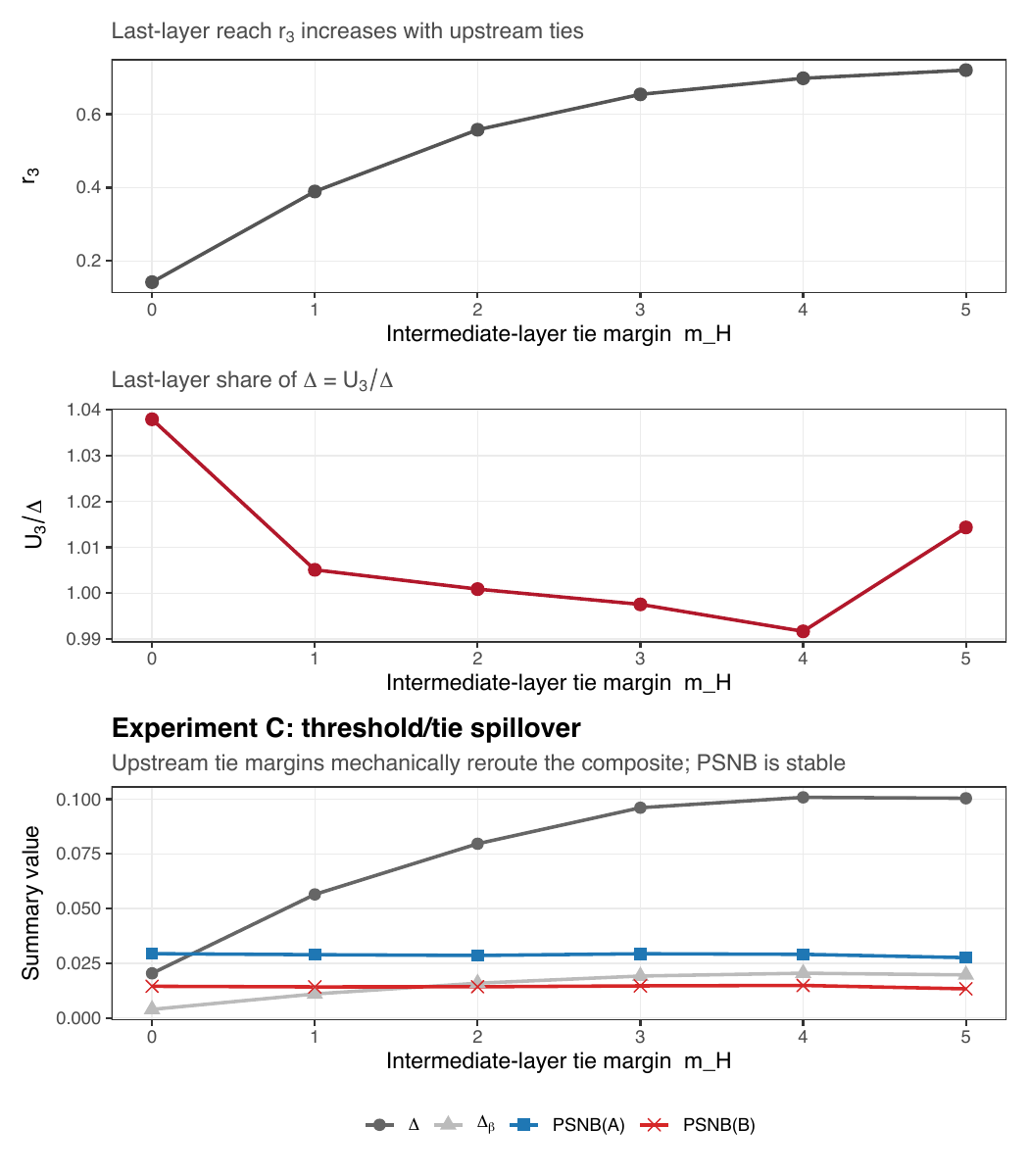}
\caption{Experiment C: threshold- and tie-induced spillover to the final layer.
Sweeping the intermediate-layer tie margin $m_H\in\{0,\ldots,5\}$ progressively
reroutes pairs to Layer~3. The panels show: (i) the final-layer reach $r_3$ rising
from $0.142$ to $0.720$; (ii) the last-layer share of standard net benefit
$U_3/\Delta$ remaining near $1$; (iii) $\Delta$, $\Delta_\beta$, and PSNB under
charters $\bm{\alpha}^{(A)}$ and $\bm{\alpha}^{(B)}$ as a function of $m_H$.
PSNB is stable, while the standard summary is driven almost entirely by the
upstream tie rule.}
\label{fig:spillover}
\end{figure}

\begin{table}[ht]
\centering
\caption{Experiment C: upstream tie spillover. Sweeping the intermediate-layer tie margin $m_H$ routes more pairs to Layer~3, shifting the standard summary $\Delta$ and the last-layer share $U_3/\Delta$ toward Layer~3 even though the true Layer-3 effect is unchanged. PSNB is stable.}
\label{tab:spillover}
\begin{tabular}{@{}llllllll@{}}
\toprule
$m_H$ & $r_3$ & $U_3/\Delta$ & $\Delta$ & $\Delta_\beta$ & PSNB(A) & PSNB(B) & $n_{\mathrm{rep}}$ \\
\midrule
0 & 0.142 & 1.038 & 0.020 & 0.004 & 0.029 & 0.014 & 1000 \\
1 & 0.389 & 1.005 & 0.056 & 0.011 & 0.029 & 0.014 & 1000 \\
2 & 0.558 & 1.001 & 0.080 & 0.016 & 0.029 & 0.014 & 1000 \\
3 & 0.654 & 0.998 & 0.096 & 0.019 & 0.029 & 0.015 & 1000 \\
4 & 0.698 & 0.992 & 0.101 & 0.020 & 0.029 & 0.015 & 1000 \\
5 & 0.720 & 1.014 & 0.100 & 0.020 & 0.028 & 0.013 & 1000 \\
\bottomrule
\end{tabular}
\end{table}

\paragraph{Take-away.}
This experiment isolates a practical phenomenon emphasized in the discussion section:
interpretability problems are not only about the last layer itself. They can also be created
upstream when thresholds or margins generate more ties and therefore mechanically reroute the
hierarchy toward the final endpoint. A purely unweighted hierarchical summary is quantitatively
at the mercy of the tie rule chosen for the intermediate layer. A charter-based PSNB summary
is not---its value is determined by the prespecified clinical priorities, not by a design
choice that has nothing to do with clinical meaning.

\subsection{Experiment D: Open-label bias stress test}

Figure~\ref{fig:bias} and Table~\ref{tab:bias_rejection} report empirical
rejection probabilities under increasing treated-arm Layer-3 bias, in two
final-layer reach regimes ($r_3\approx 0.14$ in the low-reach scenario and
$r_3\approx 0.25$ in the high-reach scenario). Both regimes use $n_1=n_0=175$
and $1500$ replicates per cell. At $b=0$ all methods reject near the nominal
$0.05$ level. As the bias shift grows, rejection increases for all methods, but
the ranking of the methods is informative.

The standard Win Ratio becomes markedly more sensitive as the biased later layer
is reached more often: at $b=10$, its rejection rate increases from $0.196$ in
the low-reach regime to $0.449$ in the high-reach regime. By contrast, PSNB's
bias sensitivity is governed primarily by the charter rather than by reach.
For charter $\bm{\alpha}^{(A)}=(0.50,0.30,0.20)$, the rejection rate is
$0.877$ in the low-reach regime and $0.885$ in the high-reach regime at $b=10$;
for charter $\bm{\alpha}^{(B)}=(0.60,0.30,0.10)$, the corresponding rates are
$0.316$ and $0.319$. Thus, once the charter is fixed, PSNB's sensitivity to
Layer-3 bias is nearly invariant to the reach regime, whereas the Win Ratio's
sensitivity changes sharply with reach.

Reducing the Layer-3 charter weight from $0.20$ to $0.10$ lowers PSNB's
rejection probability by roughly a factor of two to three at every bias level in
both reach regimes. For example, in the high-reach regime at $b=5$, the
rejection rate is $0.350$ for PSNB$(\bm{\alpha}^{(A)})$ and $0.115$ for
PSNB$(\bm{\alpha}^{(B)})$; at $b=10$, the corresponding rates are $0.885$ and
$0.319$. The same pattern appears in the low-reach regime ($0.345$ versus
$0.119$ at $b=5$, and $0.877$ versus $0.316$ at $b=10$), numerically confirming
the layer-influence cap of Corollary~\ref{cor:cap}.

The bias-budget charters make the same point more prescriptively. Charter~C,
which caps Layer~3 at $5\%$, keeps rejection close to nominal at the prespecified
$b=5$ benchmark ($0.059$ in the high-reach regime and $0.079$ in the low-reach
regime). Charter~D, which caps Layer~3 at $2\%$, keeps rejection close to nominal
even at the $b=10$ compounded-bias benchmark ($0.053$ and $0.061$, respectively).
These results illustrate how a bias budget can be translated into a charter cap
before unblinding.

Practically, PSNB should not be read as automatically attenuating late-layer bias
relative to standard reach-weighted summaries. It makes that sensitivity a function
of the charter chosen before unblinding.
Whether PSNB is less or more sensitive than the Win Ratio depends on how the
chosen Layer-3 charter weight compares with the effective influence that the
same layer would have under the standard reach-weighted analysis.

\begin{figure}[ht]
\centering
\includegraphics[width=0.95\linewidth]{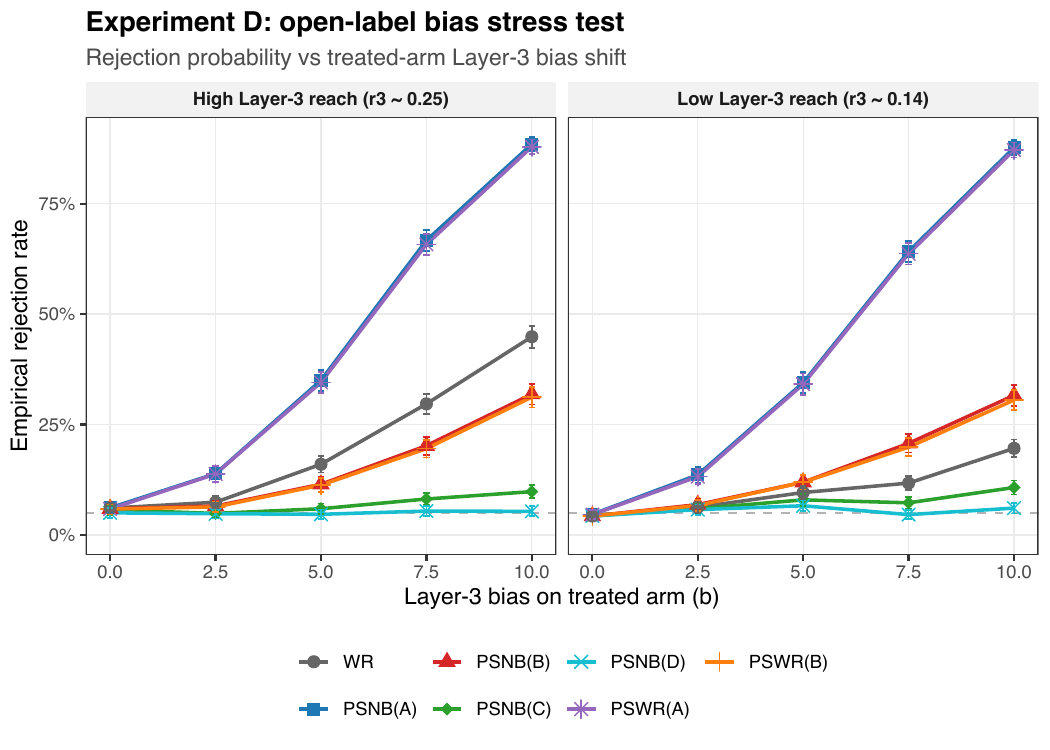}
\caption{Experiment D: spurious rejection under treated-arm Layer-3 bias.
Empirical two-sided rejection rate at nominal $\alpha = 0.05$ as a function of
the bias shift $b$, faceted by reach regime ($r_3 \approx 0.14$ in low-reach,
$r_3 \approx 0.25$ in high-reach). The Win Ratio's bias sensitivity grows with
the reach regime, whereas PSNB's sensitivity is largely governed by its
Layer-3 charter weight.}
\label{fig:bias}
\end{figure}

\begin{table}[ht]
\centering
\caption{Experiment D: open-label bias stress test. Empirical two-sided rejection rates at nominal $\alpha = 0.05$ as a function of the treated-arm Layer-3 bias shift $b$, in both a low-reach regime ($r_3 \approx 0.14$) and a high-reach regime ($r_3 \approx 0.25$). Each cell reports rejection rate (Monte Carlo SE). Smaller Layer-3 charter weights attenuate spurious rejection under measurement bias; the bias-budget charters PSNB(C) and PSNB(D) show how prespecified Layer-3 caps can limit transmission of 5- and 10-point bias shifts, respectively.}
\label{tab:bias_rejection}
\begin{tabular}{@{}lllllll@{}}
\toprule
regime & $b$ & WR & PSNB(A) & PSNB(B) & PSNB(C) & PSNB(D) \\
\midrule
high\_reach & 0.000 & 0.061 (0.006) & 0.061 (0.006) & 0.059 (0.006) & 0.057 (0.006) & 0.050 (0.006) \\
high\_reach & 2.500 & 0.074 (0.007) & 0.139 (0.009) & 0.065 (0.006) & 0.049 (0.006) & 0.048 (0.006) \\
high\_reach & 5.000 & 0.160 (0.009) & 0.350 (0.012) & 0.115 (0.008) & 0.059 (0.006) & 0.047 (0.005) \\
high\_reach & 7.500 & 0.297 (0.012) & 0.667 (0.012) & 0.202 (0.010) & 0.081 (0.007) & 0.054 (0.006) \\
high\_reach & 10.000 & 0.449 (0.013) & 0.885 (0.008) & 0.319 (0.012) & 0.098 (0.008) & 0.053 (0.006) \\
low\_reach & 0.000 & 0.043 (0.005) & 0.047 (0.005) & 0.043 (0.005) & 0.044 (0.005) & 0.043 (0.005) \\
low\_reach & 2.500 & 0.063 (0.006) & 0.136 (0.009) & 0.069 (0.007) & 0.059 (0.006) & 0.057 (0.006) \\
low\_reach & 5.000 & 0.096 (0.008) & 0.345 (0.012) & 0.119 (0.008) & 0.079 (0.007) & 0.066 (0.006) \\
low\_reach & 7.500 & 0.117 (0.008) & 0.641 (0.012) & 0.207 (0.010) & 0.073 (0.007) & 0.046 (0.005) \\
low\_reach & 10.000 & 0.196 (0.010) & 0.877 (0.008) & 0.316 (0.012) & 0.107 (0.008) & 0.061 (0.006) \\
\bottomrule
\end{tabular}
\end{table}

\paragraph{Interpretation.}
This experiment refines the main robustness claim of the paper. PSNB does not
``fix'' bias in the last layer, nor does it guarantee attenuation relative to the
standard Win Ratio. It changes the source of that sensitivity from the empirical
reach distribution to the charter.
In the present simulations, the Win Ratio's bias sensitivity increases sharply as
$r_3$ increases, whereas PSNB's bias sensitivity is essentially unchanged across
reach regimes once the charter is fixed. Smaller later-layer weights attenuate
that sensitivity; larger later-layer weights amplify it. The bias-budget charters
therefore provide a direct way to translate concern about a plausible late-layer
bias into a prespecified cap on that layer's contribution to the primary estimand.

\subsection{Experiment E: Differential missingness stress test}

Figure~\ref{fig:missingness} and Table~\ref{tab:missingness_summary} report the
differential-missingness stress test. We sweep three scenarios---balanced
(treated/control completion $0.95/0.95$), moderate ($0.74/0.90$), and severe
($0.59/0.85$)---in a high-reach Layer-3 setting where the final layer is
frequently the deciding layer. Missingness is handled by treating missing
late-layer values as ties, so that the aggregation question is isolated from
additional imputation assumptions. For each scenario we compute the bias of each
summary relative to the complete-data truth.

Across all three scenarios, the ordering is stable: the unweighted Win Ratio has
the largest absolute bias, standard net benefit is intermediate, and PSNB becomes
less sensitive as the Layer-3 charter weight is reduced. In the severe scenario,
the empirical biases are
$\widehat{\mathrm{bias}}(\Delta)=-0.019$,
$\widehat{\mathrm{bias}}(\mathrm{WR})=-0.034$,
$\widehat{\mathrm{bias}}(\mathrm{PSNB}(\bm{\alpha}^{(A)}))=-0.015$, and
$\widehat{\mathrm{bias}}(\mathrm{PSNB}(\bm{\alpha}^{(B)}))=-0.010$.
PSNB$(\bm{\alpha}^{(B)})$, which carries the smallest Layer-3 charter weight,
has roughly half the absolute bias of $\Delta$ and about one third of the bias of
the Win Ratio.

The same pattern appears in the balanced and moderate scenarios. PSNB$(\bm{\alpha}^{(B)})$
has the smallest absolute bias in every scenario, and its bias varies only modestly
across the missingness sweep (from $-0.006$ to $-0.010$), whereas the Win Ratio
remains the most sensitive summary throughout. As in Experiment~D, the important
point is not that PSNB removes missing-data bias, but that it limits how strongly
a differentially incomplete later layer can affect the \emph{primary composite}
once the charter is fixed.

\begin{figure}[ht]
\centering
\includegraphics[width=0.92\linewidth]{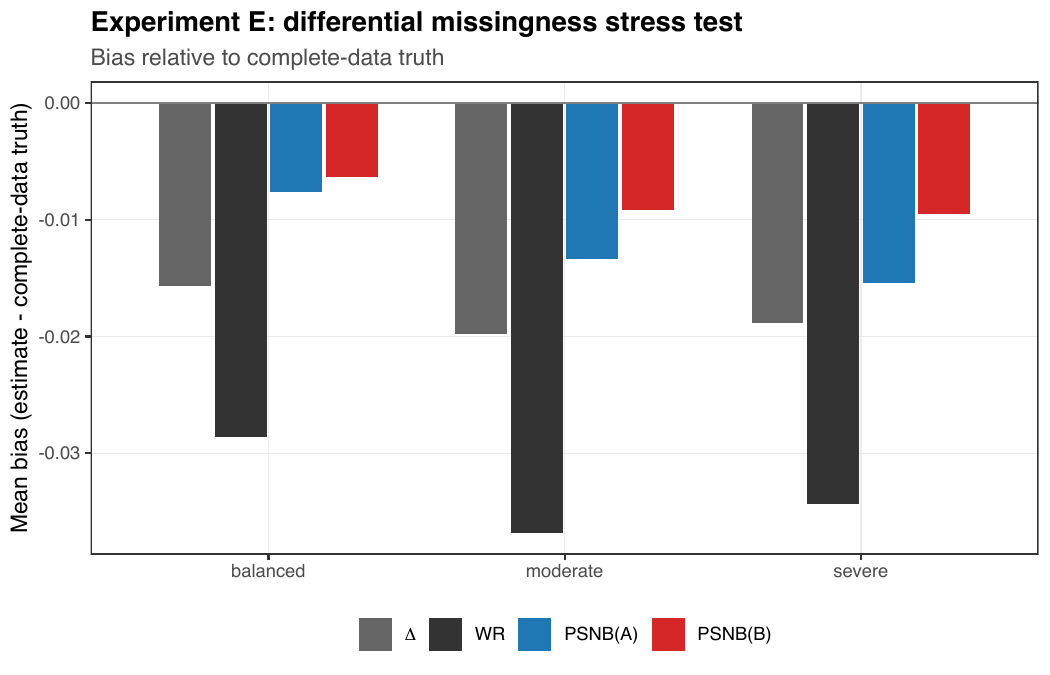}
\caption{Experiment E: differential missingness stress test. Three scenarios sweep
treated/control late-layer completion from $0.95/0.95$ (balanced) through
$0.74/0.90$ (moderate) to $0.59/0.85$ (severe). The figure shows the bias of each
summary relative to the complete-data truth. Completion rates by arm are reported
in Table~\ref{tab:missingness_summary}. PSNB$(\bm{\alpha}^{(B)})$ inherits the
smallest Layer-3 missingness sensitivity because its Layer-3 charter weight is
smallest.}
\label{fig:missingness}
\end{figure}

\begin{table}[ht]
\centering
\caption{Experiment E: differential missingness stress test. Reported are observed completion rates by arm and the mean bias (estimate minus complete-data truth) of each summary. PSNB(B), with the smallest Layer-3 weight, inherits the smallest missingness sensitivity.}
\label{tab:missingness_summary}
\begin{tabular}{@{}llllllll@{}}
\toprule
scenario & compl.\ $T$ & compl.\ $C$ & bias $\Delta$ & bias WR & bias PSNB(A) & bias PSNB(B) & $n_{\mathrm{rep}}$ \\
\midrule
balanced & 0.946 & 0.948 & -0.016 & -0.029 & -0.008 & -0.006 & 1000 \\
moderate & 0.740 & 0.897 & -0.020 & -0.037 & -0.013 & -0.009 & 1000 \\
severe & 0.589 & 0.846 & -0.019 & -0.034 & -0.015 & -0.010 & 1000 \\
\bottomrule
\end{tabular}
\end{table}

\subsection{Experiment F: Efficiency--robustness frontier and design power}

This experiment quantifies the cost of robustness \emph{within the PSNB family},
while also showing that there is no universal efficiency ordering between PSNB and
the standard Win Ratio. We sweep the upper bound $\bar\alpha_3$ on the final-layer
charter weight from $0$ to $0.30$ and report the empirical power of the resulting
capped-PSNB test against two alternative profiles: a \emph{broad-benefit} profile,
in which a moderate effect is present in every layer, and a
\emph{final-layer-dominated} profile, in which only Layer~3 carries an effect.
The unweighted Win Ratio is reported as a baseline for comparison and does not
depend on $\bar\alpha_3$. This cap-frontier portion uses $n_1=n_0=175$ and
$1000$ replicates per profile.

Figure~\ref{fig:frontier} and Table~\ref{tab:frontier_power} display the result.
Under the broad-benefit profile, capped-PSNB power rises from $0.424$ at
$\bar\alpha_3=0$ to $0.560$ at $\bar\alpha_3=0.25$, then flattens slightly.
The Win Ratio baseline is $0.507$ across the sweep, so capped PSNB crosses the
Win Ratio near $\bar\alpha_3\approx 0.10$ and exceeds it for larger caps.
Under the final-layer-dominated profile, the trade-off is sharper: capped-PSNB
power rises from $0.047$ at $\bar\alpha_3=0$ to $0.447$ at $\bar\alpha_3=0.30$,
whereas the Win Ratio baseline remains at $0.075$. A charter that completely
excludes Layer~3 is therefore essentially powerless against an alternative whose
signal is concentrated in that layer, while progressively looser caps recover
sensitivity.

The frontier is therefore not a simple monotone penalty relative to the standard
Win Ratio. Rather, it reflects how the chosen charter aligns with the
stage-conditional effect structure. Tight caps protect against later-layer
dominance and bias transmission, but can sharply reduce power against
later-layer-dominated alternatives. Looser caps recover that power and, in the
present simulations, can exceed the power of the standard Win Ratio because PSNB
allocates influence according to a prespecified stage-standardized rule rather
than according to the empirical reach distribution.

\begin{figure}[ht]
\centering
\includegraphics[width=0.92\linewidth]{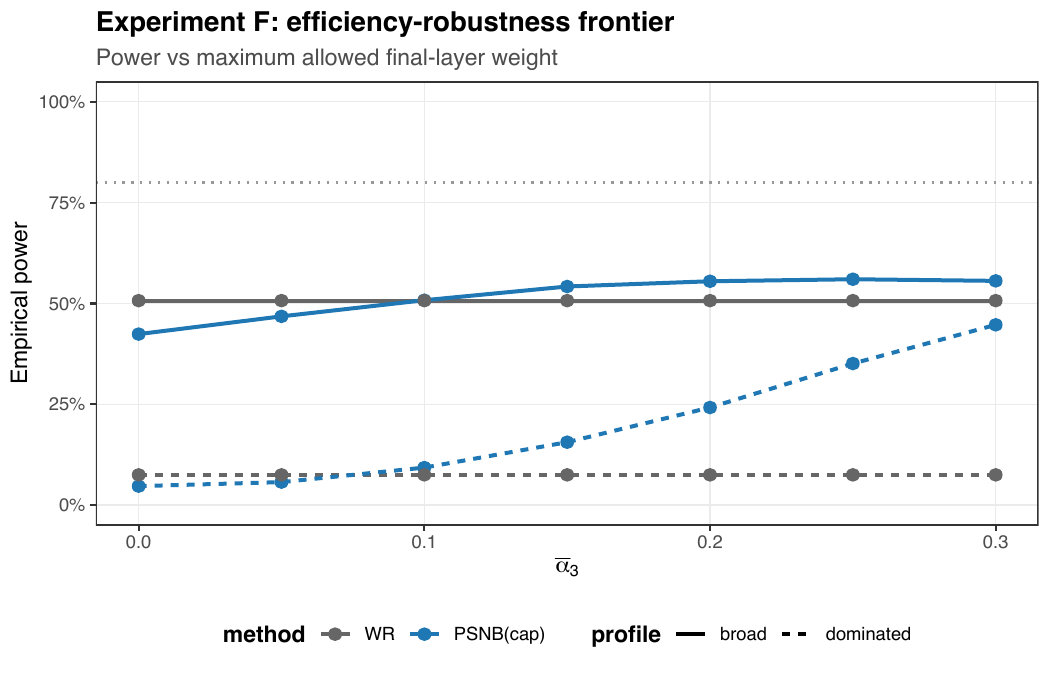}
\caption{Experiment F: efficiency--robustness frontier. Empirical power of the
capped-PSNB test against a broad-benefit alternative and a final-layer-dominated
alternative as the cap $\bar\alpha_3$ on the Layer-3 charter weight increases.
The Win Ratio baseline is constant in $\bar\alpha_3$. The figure shows the
power cost of tighter or looser Layer-3 caps.}
\label{fig:frontier}
\end{figure}

\begin{table}[ht]
\centering
\caption{Experiment F: efficiency--robustness frontier at $n=175$ per arm. Empirical power as the upper bound $\bar\alpha_3$ on the final-layer charter weight increases, under a broad-benefit profile and a final-layer-dominated profile. The trade-off between the two profiles is explicit and tunable at the charter level.}
\label{tab:frontier_power}
\begin{tabular}{@{}llllll@{}}
\toprule
$\bar\alpha_3$ & profile & method & power & MC SE & $n_{\mathrm{rep}}$ \\
\midrule
0.000 & broad & PSNB(cap) & 0.424 & 0.016 & 1000 \\
0.050 & broad & PSNB(cap) & 0.468 & 0.016 & 1000 \\
0.100 & broad & PSNB(cap) & 0.508 & 0.016 & 1000 \\
0.150 & broad & PSNB(cap) & 0.542 & 0.016 & 1000 \\
0.200 & broad & PSNB(cap) & 0.555 & 0.016 & 1000 \\
0.250 & broad & PSNB(cap) & 0.560 & 0.016 & 1000 \\
0.300 & broad & PSNB(cap) & 0.556 & 0.016 & 1000 \\
0.000 & dominated & PSNB(cap) & 0.047 & 0.007 & 1000 \\
0.050 & dominated & PSNB(cap) & 0.057 & 0.007 & 1000 \\
0.100 & dominated & PSNB(cap) & 0.093 & 0.009 & 1000 \\
0.150 & dominated & PSNB(cap) & 0.156 & 0.011 & 1000 \\
0.200 & dominated & PSNB(cap) & 0.242 & 0.014 & 1000 \\
0.250 & dominated & PSNB(cap) & 0.351 & 0.015 & 1000 \\
0.300 & dominated & PSNB(cap) & 0.447 & 0.016 & 1000 \\
0.000 & broad & WR & 0.507 & 0.016 & 1000 \\
0.050 & broad & WR & 0.507 & 0.016 & 1000 \\
0.100 & broad & WR & 0.507 & 0.016 & 1000 \\
0.150 & broad & WR & 0.507 & 0.016 & 1000 \\
0.200 & broad & WR & 0.507 & 0.016 & 1000 \\
0.250 & broad & WR & 0.507 & 0.016 & 1000 \\
0.300 & broad & WR & 0.507 & 0.016 & 1000 \\
0.000 & dominated & WR & 0.075 & 0.008 & 1000 \\
0.050 & dominated & WR & 0.075 & 0.008 & 1000 \\
0.100 & dominated & WR & 0.075 & 0.008 & 1000 \\
0.150 & dominated & WR & 0.075 & 0.008 & 1000 \\
0.200 & dominated & WR & 0.075 & 0.008 & 1000 \\
0.250 & dominated & WR & 0.075 & 0.008 & 1000 \\
0.300 & dominated & WR & 0.075 & 0.008 & 1000 \\
\bottomrule
\end{tabular}
\end{table}

The cap frontier is only one design view. Table~\ref{tab:power_by_n} reports the
corresponding sample-size sweep for the two fixed baseline charters. Under the
broad-benefit profile, all three methods gain power rapidly with increasing sample
size. At $n_1=n_0=500$, the standard Win Ratio has $0.910$ power, PSNB under
Charter~B has $0.908$ power, and PSNB under Charter~A has $0.944$ power. Thus, at a
sample size that powers the standard Win Ratio to approximately $90\%$, the baseline
PSNB charters keep pace when benefit is distributed across layers.

Under the final-layer-dominated profile, the comparison is qualitatively different.
The standard Win Ratio remains poorly powered across the examined sample-size range
($0.134$ at $n_1=n_0=500$ and $0.146$ at $n_1=n_0=800$). PSNB gains power only to the
extent that the charter allows Layer~3 to contribute: at $n_1=n_0=500$, Charter~A
reaches $0.598$ power, while the more conservative Charter~B reaches $0.190$. This
is not evidence that PSNB uniformly improves power. It shows that power is a
charter-level design consequence: a larger Layer-3 cap can recover sensitivity to a
true Layer-3 signal, whereas a smaller cap deliberately sacrifices that sensitivity
to limit bias transmission.

\begin{table}[ht]
\centering
\caption{Experiment F: simulation-based design power by per-arm sample size. Entries are empirical power with Monte Carlo standard error in parentheses, using $500$ replicates per row. The broad-benefit profile has signal across all three layers; the final-layer-dominated profile has signal only in Layer~3.}
\label{tab:power_by_n}
\begin{tabular}{@{}lllll@{}}
\toprule
profile & $n$ per arm & WR & PSNB(A) & PSNB(B) \\
\midrule
broad & 175 & 0.498 (0.022) & 0.572 (0.022) & 0.514 (0.022) \\
broad & 250 & 0.668 (0.021) & 0.710 (0.020) & 0.674 (0.021) \\
broad & 350 & 0.786 (0.018) & 0.832 (0.017) & 0.792 (0.018) \\
broad & 500 & 0.910 (0.013) & 0.944 (0.010) & 0.908 (0.013) \\
broad & 650 & 0.974 (0.007) & 0.988 (0.005) & 0.970 (0.008) \\
broad & 800 & 0.990 (0.004) & 0.996 (0.003) & 0.990 (0.004) \\
final-layer dominated & 175 & 0.074 (0.012) & 0.200 (0.018) & 0.104 (0.014) \\
final-layer dominated & 250 & 0.088 (0.013) & 0.314 (0.021) & 0.126 (0.015) \\
final-layer dominated & 350 & 0.096 (0.013) & 0.406 (0.022) & 0.128 (0.015) \\
final-layer dominated & 500 & 0.134 (0.015) & 0.598 (0.022) & 0.190 (0.018) \\
final-layer dominated & 650 & 0.146 (0.016) & 0.698 (0.021) & 0.218 (0.018) \\
final-layer dominated & 800 & 0.146 (0.016) & 0.754 (0.019) & 0.242 (0.019) \\
\bottomrule
\end{tabular}
\end{table}

\paragraph{Interpretation.}
For planning purposes, the relevant question is not only how much influence the
last layer will have under the observed reach distribution. It is also how much
influence the trial is willing to assign to that layer, and what power consequences
follow from that choice. Figure~\ref{fig:frontier} and
Table~\ref{tab:power_by_n} are quantitative versions of that question. They also
show that PSNB should not be thought of only as a robustness penalty relative to the
standard Win Ratio: depending on how the charter aligns with the stage-conditional
effect structure, PSNB can have lower or higher power than WR. These operating
characteristics are design-specific and should not be read as generic sample-size
recommendations for PSNB.

\subsection{Experiment G: Bias-budget calibration}
\label{sec:exp_g}

Figure~\ref{fig:bias_budget} traces the analytical bias budget. As the expected bias
$b$ increases, the maximum admissible last-layer weight $\bar\alpha_3$ falls, and it
falls faster for tighter tolerances $\tau$. Reading the reference charters against the
curve is informative: the more conservative baseline Charter~B $(\alpha_3=0.10)$ lies
near the $\tau\in[0.025,0.05]$ budget around $b=5$, whereas the moderate baseline
Charter~A $(\alpha_3=0.20)$ exceeds every displayed budget once the expected bias grows
beyond roughly $3$ points. The curve is a compact design device: it converts a
prespecified statement about plausible late-layer bias into an upper bound on how much
weight the last layer may carry.

\begin{figure}[ht]
\centering
\includegraphics[width=0.9\linewidth]{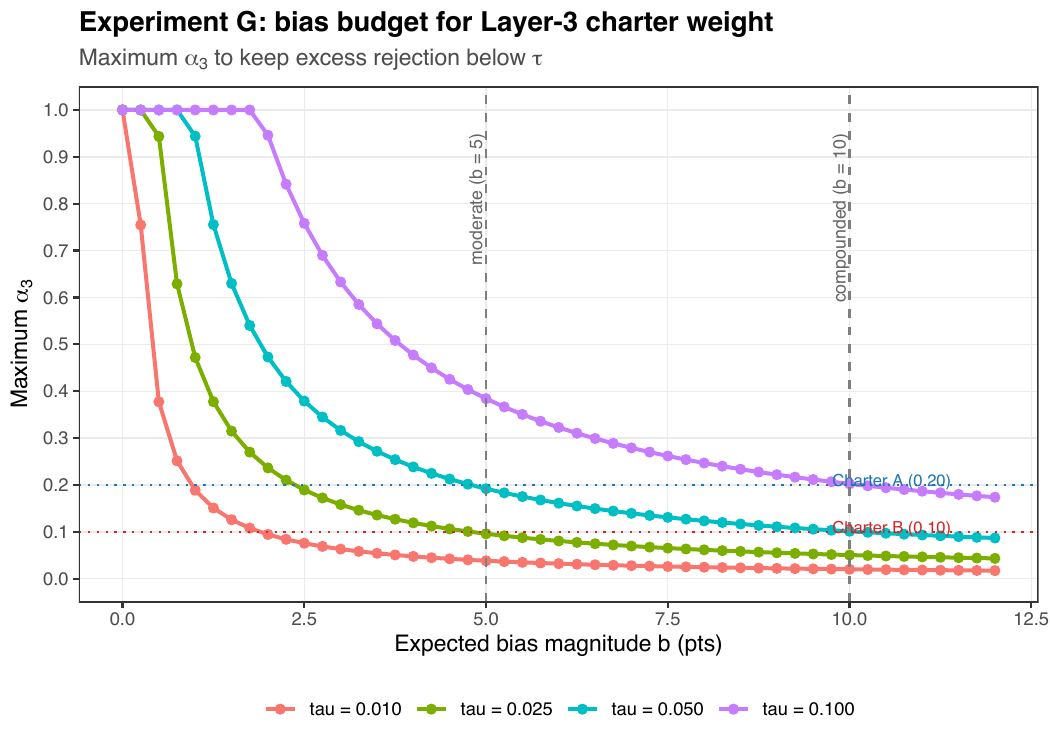}
\caption{Experiment~G, analytical bias budget. For each estimand-scale tolerance
$\tau$, the maximum last-layer charter weight $\bar\alpha_3=\tau/d_3(b)$ that keeps the
bias contribution to PSNB below $\tau$, as a function of the expected treated--control
late-layer bias $b$ (with $\delta=5$, $\sigma=10$). Dotted horizontal lines mark the two
baseline charters; dashed vertical lines mark the moderate $(b=5)$ and compounded
$(b=10)$ design benchmarks used for Charters~C and~D.}
\label{fig:bias_budget}
\end{figure}

Figure~\ref{fig:charter_calibration} reports the companion calibration simulation, which
checks the budget against realized operating characteristics. Under the global null with
no bias $(b=0)$, the empirical rejection rate stays within $[0.045,0.055]$ across the
entire $\alpha_3$ grid, confirming that the last-layer weight does not by itself disturb
calibration. As the bias grows, rejection increases monotonically in both $\alpha_3$ and
$b$. At $b=5$ the rejection rate rises from $0.057$ at $\alpha_3=0.02$ to $0.105$ at
$\alpha_3=0.10$ and $0.343$ at $\alpha_3=0.20$; at $b=10$ it rises from $0.062$ at
$\alpha_3=0.02$ to $0.282$ at $\alpha_3=0.10$ and $0.870$ at $\alpha_3=0.20$. The small
weights used by the bias-budget charters keep rejection close to nominal even under
substantial bias: at $b=10$, $\alpha_3=0.05$ (Charter~C) gives $0.090$ and
$\alpha_3=0.02$ (Charter~D) gives $0.062$, whereas $\alpha_3\ge 0.15$ inflates sharply.

\begin{figure}[ht]
\centering
\includegraphics[width=0.95\linewidth]{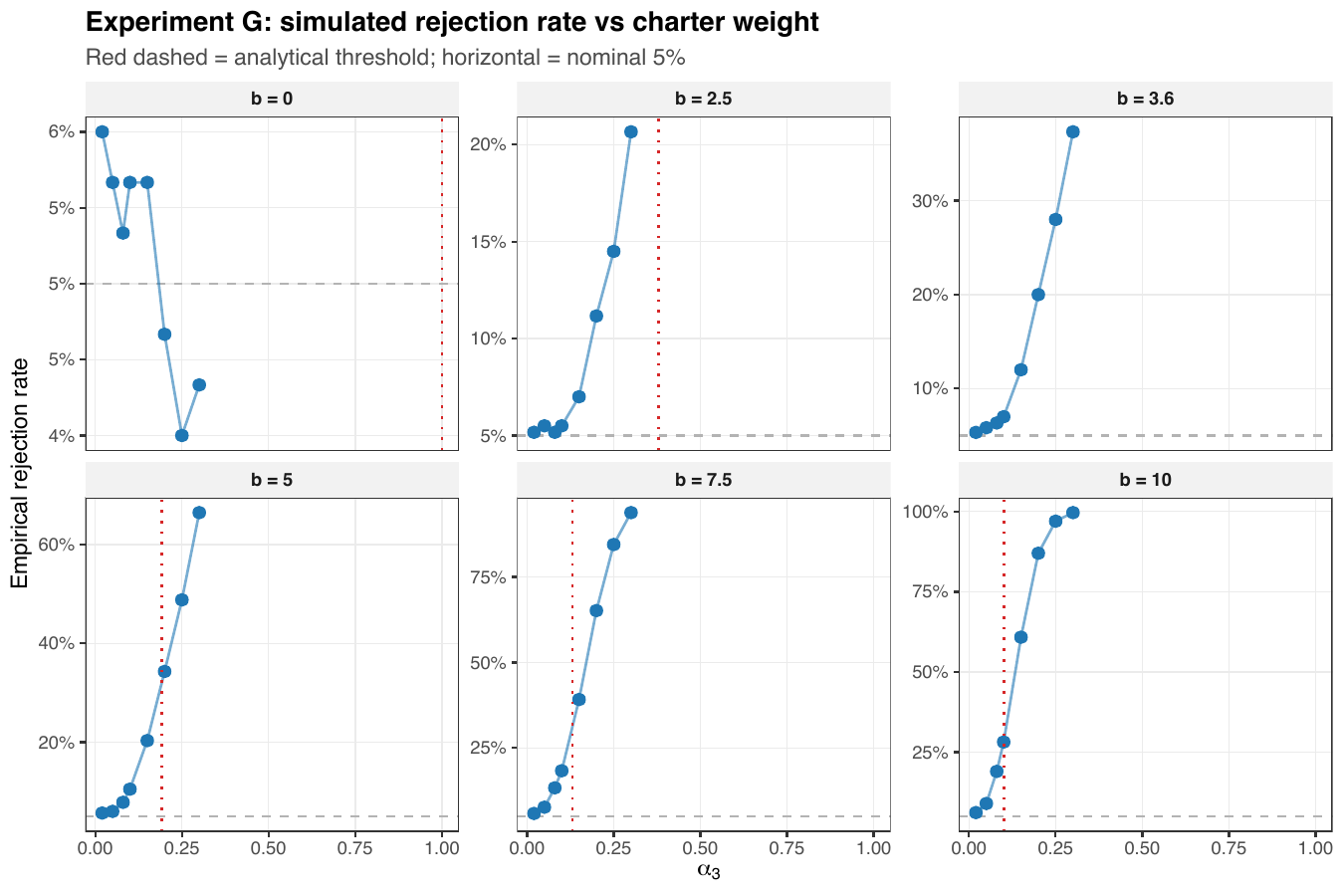}
\caption{Experiment~G, calibration simulation. Empirical two-sided rejection rate at
nominal $\alpha=0.05$ as a function of the last-layer weight $\alpha_3$, faceted by the
late-layer bias $b$, under the global null ($n_1=n_0=175$, $600$ replicates per cell).
The horizontal dashed line is the nominal level; the vertical dotted line in each panel
is the analytical $\bar\alpha_3$ at $\tau=0.05$.}
\label{fig:charter_calibration}
\end{figure}

\paragraph{Take-away.}
The analytical bias budget and the calibration simulation deliver the same qualitative
prescription---smaller late-layer weights buy robustness to late-layer bias---and let a
trial team translate a prespecified bias magnitude into a charter cap before unblinding.
The two are related but not identical: the budget bounds the estimand-scale contribution
of bias, while the simulation reports the realized type I error, and in these scenarios
the simulation supports somewhat smaller weights at moderate bias than the $\tau=0.05$
analytical bound alone would suggest. The analytical $\bar\alpha_3$ should therefore be
treated as a design starting point and confirmed against the simulated rejection rate for
the planned design, exactly as in the reporting workflow of Section~\ref{sec:report}.

\section{Practical use in trial design and reporting}

\subsection{Estimand specification}

PSNB fits naturally within the ICH E9(R1) estimand framework \citep{ICHE9R1}. A complete
specification should state:
\begin{enumerate}[leftmargin=*]
\item the target population;
\item the hierarchical outcome vector and within-layer comparison rules;
\item the strategy for intercurrent events and censoring;
\item the summary measure, namely PSNB under a prespecified charter.
\end{enumerate}
The charter and any sensitivity set should be included in the same estimand
specification, because they determine the across-layer summary being tested.

\subsection{A prospective charter workflow}

We recommend the following design-stage workflow.

\paragraph{Step 1: specify the hierarchy.}
List the layers and define all within-layer comparison rules, including any margins, thresholds,
or fixed follow-up horizon.

\paragraph{Step 2: assign baseline clinical-priority weights.}
These should reflect the scientific objective, not the expected reach pattern. The
SAP should identify the decision group and evidence used to set these weights,
including patient input when the hierarchy includes symptoms, function, or quality
of life.

\paragraph{Step 3: apply credibility modifiers.}
If a layer is more vulnerable to bias or missingness than others, reduce its weight before
renormalization. If it is objective and central to the trial objective, there is no reason to
downweight it. The modifier should be justified from prespecified evidence such as
blinding, endpoint validation, expected completeness, prior trial experience, or
external information about contextual response.

\paragraph{Step 4: prespecify a primary charter and an admissible sensitivity set.}
The admissible set may take the form
\[
\mathcal{A}
=
\left\{
\bm{\alpha}:
\alpha_k\ge 0,\ \sum_k\alpha_k=1,\ \alpha_K\le \bar\alpha_K,\ 
\alpha_1\ge\alpha_2\ge\cdots\ge\alpha_K
\right\}.
\]
This immediately yields an envelope analysis through
Proposition~\ref{prop:envelope}.
Only the primary charter should define the primary confirmatory test unless a
multiplicity adjustment for multiple charters is prespecified.

\paragraph{Step 5: report the tipping point.}
For a late layer of special concern, report the value of $\lambda^\star$ from
Proposition~\ref{prop:tipping}. This converts a broad philosophical concern into a directly
interpretable design statement.

\paragraph{Step 6: evaluate design power under plausible effect profiles.}
Before finalizing the charter, evaluate operating characteristics under at least two
profiles: a broad-benefit profile consistent with the trial objective and a
layer-dominated profile in which the most bias-sensitive layer carries most or all of
the treatment effect. This check should be performed at the proposed sample size and,
when feasible, across a sample-size grid. The goal is not to optimize the charter
post hoc, but to document the power consequence of the prespecified robustness choice.

\paragraph{Practical calibration against plausible reach.}
The simulation results suggest an additional design-stage check. Later-layer
sensitivity under PSNB is governed by the chosen charter weight, whereas later-layer
sensitivity under standard reach-weighted summaries is governed by the empirical
reach distribution. Accordingly, candidate charters should be reviewed alongside
plausible operating characteristics for reach. If the proposed charter assigns more
weight to a later layer than that layer is expected to exert under a standard
reach-weighted analysis, then the trial is, by design, choosing to emphasize that
layer more strongly; if it assigns less, then the charter is intentionally capping
that influence.

\subsection{What to report}
\label{sec:report}

For a trial report to be interpretable, we recommend reporting:
\begin{enumerate}[leftmargin=*]
\item stage reach probabilities $\widehat r_k$;
\item stage-conditional effects $\widehat\Delta_k$ with confidence intervals;
\item the primary PSNB estimate and confidence interval;
\item PSWR as a secondary interpretive summary;
\item the charter itself, including its rationale;
\item simulation-based power under the proposed charter and plausible alternative profiles;
\item a tipping-point analysis for the most sensitive late layer; and
\item a charter-envelope sensitivity analysis when the weight set is not point-identified by
the protocol rationale alone.
\end{enumerate}

Table~\ref{tab:reporting_template} gives a compact SAP and reporting template.

\begin{table}[ht]
\centering
\small
\caption{Practical PSNB specification and reporting template for a confirmatory trial.}
\label{tab:reporting_template}
\begin{tabular}{p{0.22\linewidth}p{0.35\linewidth}p{0.31\linewidth}}
\toprule
Item & SAP specification & Trial report output \\
\midrule
Hierarchy and rules &
Layer order, margins or thresholds, fixed horizon $\tau$, tie rules, and censoring rule. &
Observed reach $\widehat r_k$ and stage-conditional effects $\widehat\Delta_k$ by layer. \\
Charter &
Primary $\bm{\alpha}$, rationale, approval process, and any caps or monotone constraints. &
Primary PSNB estimate, confidence interval, and test under the prespecified charter. \\
Censoring and missingness &
Restricted-horizon, IPCW, hypothetical, or sensitivity strategy for incomplete pairwise comparisons. &
Completeness by layer and prespecified sensitivity analyses for censoring or missingness. \\
Operating characteristics &
Type I error checks, design power, and bias or missingness stress tests under plausible profiles. &
Design-specific power results and the profiles used to interpret robustness. \\
Sensitivity charters &
Sensitivity menu or convex admissible set, plus multiplicity status. &
Tipping point, charter envelope, and clear separation of confirmatory and supportive results. \\
\bottomrule
\end{tabular}
\end{table}

The point of the template is traceability. A reader should be able to determine
which charter controlled the primary hypothesis, which analyses were supportive,
and whether later-layer influence was limited by design or allowed because it was
central to the trial objective.

\paragraph{Illustrative SAP specification.}
For example, consider a 1:1 randomized cardiovascular or structural-heart trial
with a three-layer hierarchy: death through 12 months, nonfatal disease-related
hospitalization through 12 months, and a 12-month health-status change score
compared using a 5-point margin. The SAP
could define a restricted-horizon rule for the time-to-event layers, treat pairs
not orderable before the horizon according to a prespecified censoring rule, and
set the primary charter to $(0.57,0.38,0.05)$ after documenting that the final
layer is clinically important but more vulnerable to expectation bias and
differential missingness. The primary hypothesis test would be the two-sided Wald
test of $\Delta_{\mathrm{PS}}(\bm{\alpha})=0$ under that charter; PSWR would be
reported as a secondary ratio-scale summary.

The same SAP would prespecify a sensitivity set, for example
$\alpha_3\le 0.10$ with $\alpha_1\ge\alpha_2\ge\alpha_3$, and would state that the
charter envelope, tipping point, and alternative charters are supportive unless a
multiplicity adjustment is specified. The trial report would then present
$\widehat r_k$, $\widehat\Delta_k$, the primary PSNB estimate and confidence
interval, the design-power profiles used to justify the sample size, and the
preplanned sensitivity analyses for missingness and censoring.

\subsection{What PSNB does not solve}

The scope of PSNB is limited. It does not remove bias within a layer, solve missing-data problems,
or replace principled handling of censoring. It also does not make the charter
objective or self-evident; the charter remains a clinical and regulatory judgment
that must be justified. Those issues must still be addressed in the trial's
estimand and analysis plan. What PSNB does is narrower and, in our view, important:
it keeps the \emph{overall composite summary} from giving a frequently reached late
layer full numerical influence unless that role has been prespecified and justified.
The simulations in this paper do not validate IPCW, informative-censoring, or
missing-data models; they evaluate how the across-layer aggregation rule behaves
after the layer-specific comparison and incomplete-data rules have been chosen.

\section{Discussion}

The issue addressed here is not only how to compare treated and control participants
\emph{within} each layer, but how to combine those layer-specific comparisons
\emph{across} the hierarchy. Standard win statistics use the data-generated reach
distribution for that aggregation. PSNB replaces those stochastic weights with
charter weights fixed before unblinding.

That change has three practical consequences.

First, it sharpens interpretation. PSNB is a summary of stage-conditional pairwise
effects under a prespecified layer-importance scheme. It is therefore much closer
to the language investigators use when they say that some layers are more important
or more credible than others.

Second, it makes layer influence reviewable. The layer influence cap in
Corollary~\ref{cor:cap} and the tipping-point analysis in
Proposition~\ref{prop:tipping} translate an otherwise opaque concern---``the late
subjective layer may dominate''---into quantities that can be computed, reported,
and prespecified.

Third, it provides a design frontier. The simulation results show that
there is no single efficiency ordering between PSNB and the standard Win Ratio.
Within the PSNB family, reducing the weight on a highly informative final layer
lowers sensitivity to that layer but can also reduce power against alternatives
concentrated there. Relative to the Win Ratio, PSNB may have lower or higher power
depending on how the charter aligns with the stage-conditional effect structure.
In the design-power sweep, the baseline PSNB charters kept pace with a Win Ratio
powered to approximately $90\%$ under broad benefit at $n_1=n_0=500$, but the more
conservative charter had limited power when the true signal was concentrated only in
Layer~3. The gain is that the trade-off is made during design rather than inferred
only after the reach distribution is observed.

The simulations clarify these points in ways that are directly relevant to trial
practice. In Experiment~B, PSNB changed by only a few thousandths across an
order-of-magnitude sweep in final-layer reach while the standard net benefit
changed by almost $0.10$, confirming that PSNB targets stage-conditional treatment
structure rather than the empirical flow of pairs through the hierarchy. In
Experiment~C, changing only an upstream tie rule rerouted pairs to the final layer
and made the standard summary almost entirely later-layer-driven, even though the
later-layer treatment effect itself was unchanged. In Experiment~D, the standard
Win Ratio's sensitivity to later-layer bias increased sharply as reach increased,
whereas PSNB's sensitivity was nearly unchanged across reach regimes once the
charter was fixed. These results do \emph{not} show that PSNB automatically reduces
bias relative to standard summaries. They show a narrower point: the sensitivity is
controlled by the charter rather than by the realized reach distribution.

\paragraph{Relationship to adjacent methods.}
PSNB does not compete with thresholding, margins, win odds, or refined censoring
estimands. Those methods answer different questions. Thresholds and margins define
what constitutes a win within a layer. Estimand work under censoring clarifies what
pairwise target is being estimated. PSNB addresses how stage-specific information is
aggregated once those choices are made. In that sense, it is complementary to recent
methodological work rather than a replacement for it.

\paragraph{Limitations.}
The present paper emphasizes the estimand and design logic of PSNB. Several
extensions are left for future work: covariate-adjusted or doubly robust
estimation; more systematic treatment of censoring through IPCW or
restricted-horizon formulations; formal methods for eliciting charters from
multiple stakeholders; and applied case studies in completed trials. In addition,
the missingness and bias experiments reported here are stress tests designed to
expose mechanisms rather than a complete theory for all forms of measurement error
or missing-data bias. The bias-budget charters are therefore examples of how an SAP
could translate prespecified evidence into caps, not evidence that any particular
cap is generally correct.

\paragraph{Recommended use.}
We view PSNB as particularly attractive when a hierarchy contains a later layer
that is both frequently reached and more vulnerable to bias or missingness than
earlier layers. In such settings, the scientific issue is rarely whether the later
layer is important; it usually is. The issue is whether the primary composite
should allow that layer to dominate without saying so in the estimand. PSNB forces
that choice into the design.
The simulations further suggest that the later-layer charter should not be chosen
in isolation, but in light of plausible operating characteristics for reach and
power. The charter determines how much of that layer's signal---or bias---is allowed
into the primary estimand, so trial reports should make both the robustness rationale
and the corresponding power consequences visible.

\subsection*{Conclusion}

Priority-Standardized Net Benefit is an estimand for hierarchical composite
endpoints that changes only the across-layer aggregation rule. By replacing
stochastic reach weights with a prespecified priority/credibility charter, PSNB
keeps last-layer influence from being determined solely by how often pairs reach
that layer. The completed simulations show that PSNB is nearly
invariant to large changes in reach when stage-conditional effects are fixed,
that it reduces the effect of upstream tie rules on across-layer interpretation,
and that later-layer bias and missingness sensitivity follows the charter rather
than the observed reach distribution. The design-power simulations show that PSNB can keep pace
with a well-powered Win Ratio when benefit is broad, while power under
layer-dominated benefit depends directly on the prespecified Layer-3 cap. The
accompanying influence-function-based estimator
provides practical inference, and the proposed design tools---layer influence
caps, tipping-point analysis, and charter envelopes---turn debates about layer
credibility into analysis choices that can be written down before unblinding. For
trials that use hierarchical composites and worry that one later layer may dominate the result,
PSNB offers a bounded, interpretable, and practically implementable alternative.

\appendix

\section{Algorithmic summary}

The following steps compute PSNB from participant-level data.

\begin{enumerate}[leftmargin=*,label=\textbf{Step \arabic*.}]
\item For each treated--control pair, compute stage comparison scores $c_k$ and reach indicators $R_k$.
\item Estimate $\widehat U_k$ and $\widehat r_k$ for each stage.
\item Form $\widehat\Delta_k=\widehat U_k/\widehat r_k$.
\item Aggregate with the prespecified charter:
$\widehat\Delta_{\mathrm{PS}}=\sum_k\alpha_k\widehat\Delta_k$.
\item Estimate the projection variance using the treated- and control-arm
projection averages described in Section~5.3 and equation~\eqref{eq:plugin_var},
and construct Wald or bootstrap confidence intervals.
\end{enumerate}

\bibliographystyle{plainnat}

\begin{thebibliography}{99}

\bibitem[Bebu and Lachin(2016)]{BebuLachin2016}
I.~Bebu and J.~M. Lachin.
\newblock Large sample inference for a win ratio analysis of a composite outcome
  based on prioritized components.
\newblock \emph{Biostatistics}, 17(1):178--187, 2016.
\newblock \url{https://doi.org/10.1093/biostatistics/kxv032}

\bibitem[Brunner et~al.(2021)]{Brunner2021}
E.~Brunner, M.~Vandemeulebroecke, and T.~M{\"u}tze.
\newblock Win odds: An adaptation of the win ratio to include ties.
\newblock \emph{Statistics in Medicine}, 40(14):3367--3384, 2021.
\newblock \url{https://doi.org/10.1002/sim.8967}

\bibitem[Buyse(2010)]{Buyse2010}
M.~Buyse.
\newblock Generalized pairwise comparisons of prioritized outcomes in the
  two-sample problem.
\newblock \emph{Statistics in Medicine}, 29(30):3245--3257, 2010.
\newblock \url{https://doi.org/10.1002/sim.3923}

\bibitem[Dong et~al.(2020)]{Dong2020}
G.~Dong, D.~C. Hoaglin, J.~Qiu, R.~A. Matsouaka, Y.-W. Chang, J.~Wang, and M.~Vandemeulebroecke.
\newblock The win ratio: On interpretation and handling of ties.
\newblock \emph{Statistics in Biopharmaceutical Research}, 12(1):99--106, 2020.
\newblock \url{https://doi.org/10.1080/19466315.2019.1575279}

\bibitem[Even and Josse(2025)]{EvenJosse2025CausalWR}
M.~Even and J.~Josse.
\newblock Rethinking the win ratio: A causal framework for hierarchical outcome analysis.
\newblock \emph{arXiv:2501.16933}, 2025.
\newblock \url{https://arxiv.org/abs/2501.16933}

\bibitem[Finkelstein and Schoenfeld(1999)]{FinkelsteinSchoenfeld1999}
D.~M. Finkelstein and D.~A. Schoenfeld.
\newblock Combining mortality and longitudinal measures in clinical trials.
\newblock \emph{Statistics in Medicine}, 18(11):1341--1354, 1999.
\newblock \url{https://doi.org/10.1002/(SICI)1097-0258(19990615)18:11%3C1341::AID-SIM129%3E3.0.CO;2-7}

\bibitem[ICH(2019)]{ICHE9R1}
International Council for Harmonisation (ICH).
\newblock ICH E9(R1) addendum on estimands and sensitivity analysis in clinical trials to
  the guideline on statistical principles for clinical trials.
\newblock Step 4 guideline, dated 20 November 2019.
\newblock \url{https://www.ich.org/page/efficacy-guidelines}

\bibitem[Luo et~al.(2017)]{Luo2017WeightedWinLoss}
X.~Luo, J.~Qiu, S.~Bai, and H.~Tian.
\newblock Weighted win loss approach for analyzing prioritized outcomes.
\newblock \emph{Statistics in Medicine}, 36(15):2452--2465, 2017.
\newblock \url{https://doi.org/10.1002/sim.7284}

\bibitem[Mao et~al.(2022)]{Mao2022SampleSize}
L.~Mao, K.-M. Kim, and X.~Miao.
\newblock Sample size formula for general win ratio analysis.
\newblock \emph{Biometrics}, 78(3):1257--1268, 2022.
\newblock \url{https://doi.org/10.1111/biom.13501}

\bibitem[Mao(2024)]{Mao2024EstimandWR}
L.~Mao.
\newblock Defining estimand for the win ratio: separate the true effect from censoring.
\newblock \emph{Clinical Trials}, 21(5):584--594, 2024.
\newblock \url{https://doi.org/10.1177/17407745241259356}

\bibitem[Mou et~al.(2024)]{Mou2024thresholds}
Y.~Mou, T.~Kyriakides, S.~Hummel, F.~Li, and Y.~Huang.
\newblock Generalizing the Finkelstein--Schoenfeld test to incorporate multiple alternating thresholds.
\newblock \emph{arXiv:2407.18341}, 2024.
\newblock \url{https://arxiv.org/abs/2407.18341}

\bibitem[Pocock et~al.(2012)]{Pocock2012}
S.~J. Pocock, C.~A. Ariti, T.~J. Collier, and D.~Wang.
\newblock The win ratio: a new approach to the analysis of composite endpoints in clinical trials
  based on clinical priorities.
\newblock \emph{European Heart Journal}, 33(2):176--182, 2012.
\newblock \url{https://doi.org/10.1093/eurheartj/ehr352}

\bibitem[Pocock et~al.(2024)]{Pocock2024Lessons}
S.~J. Pocock, J.~Gregson, T.~J. Collier, J.~P. Ferreira, and G.~W. Stone.
\newblock The win ratio in cardiology trials: lessons learnt, new developments, and wise future use.
\newblock \emph{European Heart Journal}, 45(44):4684--4699, 2024.
\newblock \url{https://doi.org/10.1093/eurheartj/ehae647}

\end{thebibliography}

\end{document}